%% file: IntervalColoringFullVersion.tex
\documentclass[preprint,9pt]{elsarticle}
\usepackage{graphicx}
\usepackage{amsmath}
\usepackage{amssymb}
\usepackage{multirow}
\usepackage{microtype}%if unwanted, comment out or use option "draft"
\usepackage{listings}
\usepackage{url,graphics,amsmath,amsfonts,amssymb,tikz,verbatim,amsthm}
\usepackage{algorithm}
\usepackage{mathtools}
\usepackage{algpseudocode}
\usepackage{caption}
\usepackage{fullpage}

\newtheorem{theorem}{Theorem}
\newtheorem{lemma}[theorem]{Lemma}

\algrenewcomment[1]{\(\triangleright\) #1}
\newcommand{\C}{\sf \textbf{C}}
\newcommand{\D}{\sf DIRTY}
\newcommand{\R}{\sf row}

\newcommand{\out}{\sf out}
\newcommand{\KTSLS}{\sf KT\mbox{-}SLS}
\newcommand{\Insertion}{\sf handle \mbox{-} Insert}
\newcommand{\Deletion}{\sf handle \mbox{-} Delete}
\newcommand{\OFFSET}{\sf compute\mbox{-}Offset}
\newcommand{\computemaxheightSLS}{\sf compute\mbox{-} max\mbox{-}SLS}
\newcommand{\getSLS}{\sf compute\mbox{-}SLS}
\newcommand{\updateEndPoints}{\sf update\mbox{-}Endpoints}

\journal{Theoretical Computer Science}

\begin{document}
\begin{frontmatter}

\title{Dynamic Data Structures for Interval Coloring \tnoteref{conf}}
\tnotetext[conf]{Preliminary version of this work appeared in $25^{th}$ International Computing and Combinatorics Conference(COCOON),pages 478-489, 2019}

\author[first]{Girish Raguvir J}
\ead{girishraguvir@gmail.com}
\author[first]{Manas Jyoti Kashyop}
\ead{manasjk@cse.iitm.ac.in}
\author[first]{N. S. Narayanaswamy}
\ead{swamy@cse.iitm.ac.in}
\address[first]{Department of Computer Science and Engineering, Indian Institute of Technology Madras, Chennai 600036, India}
\renewcommand{\labelitemi}{$\blacksquare$}

\begin{abstract}
\input{Journal/abstract.tex}
\end{abstract}
\end{frontmatter}
% --------------------------------------------------------------------

% 
\input{Journal/NewIntroduction.tex}

\renewcommand{\labelitemi}{$\bullet$}
\input{Journal/preliminaries.tex}
\input{Journal/Incremental.tex}
\input{Journal/FullyDynamic.tex}
\input{Journal/OMV.tex}

%--------------------------------------------------------------------

\bibliography{References}
\bibliographystyle{elsarticle-num}

\end{document}

%% file: Journal/abstract.tex
% !TEX root = ../IntervalColoringFullVersion.tex
We consider the dynamic graph coloring problem restricted to the class of interval graphs in the incremental and fully dynamic setting. 
The input consists of a sequence of intervals that are to be either colored, or deleted, if previously colored.   For the incremental setting, we consider the well studied optimal online algorithm (KT-algorithm) for interval coloring due to Kierstead and Trotter  \cite{kierstead1981extremal}.  
%
%Let $n$ denotes the total number of intervals in the online update sequence, $\Delta$ denotes the maximum degree of a vertex and $\omega$ denotes the size of the maximum clique in the interval graph associated with these intervals.
%%Input to the problem is an online update sequence. The goal is to maintain a proper coloring of the intervals after every update with as few colors as possible while maintaining a small update time.\par
%In the online setting, KT-algorithm achieves the best competitive ratio. 
We  present the following results on the dynamic interval coloring problem.
%However, we show that a sub-quadratic time implementation of the KT-algorithm is unlikely to exist and present our first result which is a lower bound in the implementation of the KT-algorithm.
\begin{itemize}
\item Any {\em direct} implementation of the KT-algorithm requires $\Omega(\Delta^2)$ time per interval in the worst case.
% where $n$ is the total number of intervals in the online sequence presented to the algorithm and $\Delta$ is the maximum degree of a vertex in the interval graph formed by those intervals. 
\end{itemize}
%In the incremental setting, each update step presents the algorithm with an interval to be colored. We achieve the following result in the incremental setting.    
%We design an incremental algorithm that is subtly different from the KT-algorithm and uses at most $3 \omega - 2$ colors, where $\omega$ is the size of the maximum clique in the interval graph associated with the set of intervals.
\begin{itemize}
\item There exists an incremental algorithm which supports insertion of an interval in amortized $O(\log n + \Delta)$ time per update and maintains a proper coloring using at most $3 \omega - 2$ colors.
% where $n$ is the total number of intervals in the update sequence and $\Delta$ is the maximum degree of a vertex in the interval graph formed by those intervals.
\end{itemize}
%Next we turn our attention to the fully dynamic setting. 
%In the fully dynamic setting, at each update step the algorithm is presented with an interval to be colored, or a previously colored interval to delete. We achieve the following result in the fully dynamic setting. 
%On each update, our aim is to maintain a $3 \omega - 2$ coloring of the remaining set of intervals while maintaining the update time as small as possible, where $\omega$ is the size of the maximum clique in the interval graph associated with the remaining set of intervals. 
\begin{itemize}
\item There exists a fully dynamic algorithm which supports insertion of an interval in $O(\log n + \Delta \log \omega)$ update time and deletion of an interval in $O(\Delta^2 \log n)$ update time in the worst case and maintains a proper coloring using at most $3 \omega - 2$ colors.
% where $n$ is the total number of intervals inserted and $\Delta$ is the maximum degree of a vertex in the interval graph formed by those intervals.   
\end{itemize}
%Finally, we consider Online Boolean Matrix-Vector multiplication conjecture (OMv) which is extensively used to prove lower bound results in the dynamic setting\cite{DBLP:conf/stoc/HenzingerKNS15}. 
The KT-algorithm crucially uses  the maximum clique size in an induced subgraph in the neighborhood of a given vertex.  
We show that the problem of computing the induced subgraph among the neighbors of a given vertex has the same hardness as the online boolean matrix vector multiplication problem  \cite{DBLP:conf/stoc/HenzingerKNS15}. We show that 
\begin{itemize}
\item Any algorithm that computes the induced subgraph among the neighbors of a given vertex requires at least quadratic time unless the OMv conjecture \cite{DBLP:conf/stoc/HenzingerKNS15} is false.
\end{itemize}
Finally, we obtain the following result on the OMv conjecture.
\begin{itemize}
\item If the matrix and the vectors in the online sequence have the consecutive ones property, then the OMv conjecture \cite{DBLP:conf/stoc/HenzingerKNS15} is false.
\end{itemize}  

%% file: Journal/NewIntroduction.tex
% !TEX root = ../IntervalColoringFullVersion.tex
\section{Introduction}
Maintenance of data structures for graphs in the dynamic setting has been extensively studied. In the dynamic setting, a graph has a fixed set of vertices whereas the edge set keeps evolving by means of edge updates. An edge update consists of either insertion of a new edge or deletion of an existing edge. A dynamic graph is thus a sequence of graphs, $\mathcal{G} = \{G_0, G_1,......, G_t\}$, where $t$ is the total number of edge updates, initial graph $G_0 = (V, \phi)$ is an empty graph and graph $G_i$ is obtained from $G_{i-1}$ by a single edge update. In our work, $G_i$ is an interval graph and an update consists of an interval to be inserted or deleted. Therefore, in our dynamic setting, a single update may insert or delete many edges in the underlying interval graph. This is different from the commonly studied case in the area of dynamic graph algorithms where on each edge update a single edge is inserted or deleted. 
\par 
The graph coloring problem is one of the most extensively studied problems. In the dynamic setting, graph coloring problem is  as follows:  there is an online sequence of edge updates and the goal is to maintain proper coloring after every update. Several works (\cite{DBLP:conf/COSSOM/DutotGOP},\cite{DBLP:conf/CCCA/OuerfelliB},\cite{DBLP:SallinenIPGRP} and \cite{DBLP:conf/HardyLT}) propose heuristic and experimental results on the dynamic graph coloring problem. To the best of our knowledge, the formal analysis of data structures for dynamic graph coloring have been done in \cite{DBLP:journals/corr/HenzingerConstanttime}, \cite{DBLP:journals/corr/SayanConstanttime},  \cite{DBLP:conf/esa/SolomonW18_coloring}, \cite{DBLP:conf/soda/BhattacharyaCHN18}, \cite{DBLP:conf/iccS/BarenboimM17}, and \cite{DBLP:conf/wads/BarbaCKLRRV17}.  %Bhattacharya et al. give the current best fully dynamic randomized algorithm which maintains $\Delta + 1$ vertex coloring in $O(\log \Delta)$ expected amortized update time \cite{DBLP:conf/soda/BhattacharyaCHN18}. They also give the current best deterministic algorithm which maintains $\Delta + o(\Delta)$ vertex coloring in $O(polylog \Delta)$ amortized update time \cite{DBLP:conf/soda/BhattacharyaCHN18}. \\
We continue the study of dynamic data structures for graph coloring. We focus on interval graphs in the incremental as well as in the fully dynamic setting. The online update sequence consists of intervals and our goal is to maintain a proper coloring of the intervals with as few colors as possible while maintaining a small update time. 
In the incremental setting, each update in the online update sequence consists of an interval to be colored. In the fully dynamic setting, each update in the online update sequence consists of either an interval to be colored or a previously colored interval to be deleted.\par 
%Thus, in our dynamic setting, a single update may insert or delete many edges in the underlying interval graph. This is different from the commonly studied case in the area of dynamic graph algorithms where on each edge update a single edge is inserted or deleted.\par   
In the incremental setting,  intervals in the update sequence are inserted one after the other and we aim to efficiently maintain a proper coloring of the intervals using as few colors as possible after every update. Our approach is to consider efficient implementations of well-studied online algorithms for interval coloring.  Online algorithms for interval coloring and variants is a rich area with many results \cite{DBLP:conf/icalp/EpsteinL05}.  Note that an online algorithm is not allowed to re-color a vertex during the execution of the algorithm. On the other hand, an incremental algorithm is not restricted in anyway during an update step except that we desire that the updates be done as efficiently as possible.  Naturally, an online interval coloring algorithm which is efficiently implementable is a good candidate for an incremental interval coloring algorithm as it only assigns a color to the current interval, and does not change the color of any of the other intervals.  
For the online interval coloring problem, Kierstead and Trotter presented a $3$ competitive algorithm (KT-algorithm) and they also proved that their result is tight \cite{kierstead1981extremal}. The tightness is proved by showing the existence of an adaptive adversary that forces an online algorithm to use $3 \omega - 2$ colors where $\omega$ is the maximum clique size in the interval graph formed by the given set of intervals. On the other hand, the KT-algorithm uses at most $3 \omega - 2$ colors. 
% Towards this aim, we explore the existing works on online algorithms for interval coloring.\par 
% Epstein et al. studied online graph coloring for interval graphs \cite{DBLP:conf/icalp/EpsteinL05}. They studied four variants of the problem: online interval coloring with bandwidth, online interval coloring without bandwidth, lazy online interval coloring with bandwidth, and lazy online interval coloring without bandwidth. For online interval coloring with bandwidth, Narayanaswamy presented an algorithm with competitive ratio 10 \cite{DBLP:conf/cocoon/Narayanaswamy04} and  Epstein et  al. showed a lower bound of $3.2609$ \cite{DBLP:conf/icalp/EpsteinL05}. For lazy online interval coloring with bandwidth and lazy online interval coloring without bandwidth, Epstein et al. proved that competitive ratio can be arbitrarily bad for any online algorithm \cite{DBLP:conf/icalp/EpsteinL05}.
%Therefore, KT-algorithm has the optimum competitive ratio.  
%In other words, The online algorithm (KT-algorithm) due to Kierstead and Trotter \cite{kierstead1981extremal} is known to have the optimum competitive ratio.   
\subsection{Our Results} 
Our goal is to design incremental and fully-dynamic algorithms for interval coloring.  Towards this, we study efficient implementations of the KT-algorithm. The KT-algorithm computes a coloring in which each color is a 2-tuple $(p(v),o(v))$, where $p(v)$ is the level value of $v$ and $o(v)$ is the offset of $v$.
In the incremental and fully-dynamic setting, we design efficient 3-approximation algorithms for interval coloring.  In the incremental case our results leave open the possibility of improving the number of colors used by sacrificing the constraint in online algorithms that an interval cannot be re-colored.  
%which maintain for maintaining a coloring of intervals which 
%We start with an aim to achieve an efficient implementation of KT-algorithm. To the best of our knowledge, no study has been reported about efficient implementation of the KT-algorithm.  
%The KT-algorithm computes a proper coloring by assigning to each vertex a color which is a 2-tuple denoted by $($level,offset$)$. The value of level is in the range $[0,\omega-1]$ and the value of offset is from the set $\{1,2,3\}$. Further, all the vertices whose level value is 0 form an independent set. Therefore, 
%An implementation of the KT-algorithm gives us an incremental interval coloring algorithm that maintains a proper coloring with at most $3 \omega - 2$ colors and without any recoloring. However, our result in Section~\ref{sec:KTalgorithm-hardness} deters us from obtaining such an algorithm.\par
%%%%%%%%%%%%%%%%%%%%%%%%%%%%%%%%%%%%%%%%%%%
We start by considering the efficiency of a {\em direct} implementation of the KT-algorithm.  A direct implementation uses a data structure that only maintains the intervals and responds to intersection queries by reporting the intervals which intersect a queried interval. We show the following result in Section~\ref{sec:KTalgorithm-hardness}. 
\begin{itemize}
\item Any direct implementation of the KT-algorithm requires $\Omega(\Delta^2)$ time per interval in the worst case, where $\Delta$ is the maximum degree of a vertex in the associated interval graph. (Theorem~\ref{thm:LowerBound-KT-algorithm}) 
\end{itemize}
We then show that a  comparison based data structure which supports the insertion of a new interval or  computes the number of intervals intersecting  a given interval requires $\Omega(\log n)$ comparisons for at least one of the operations (Lemma~\ref{lem:Lowerbound-DS-Interval}).  
%The reason for this is that the level value of an interval computed by the KT-algorithm depends on the size of the maximum clique in the graph induced by the intervals intersecting with it. In Section~\ref{sec:KTalgorithm-hardness}, we show that any algorithm computing the maximum clique size formed by a given set of intervals takes linear time (Lemma~\ref{lem:LowerBound-maximumClique}).  
%In the worst case, to compute the level value for a single interval, KT-algorithm repeats the computation of the maximum clique size linear number of times. 
%%%%%%%%%%%%%%%%%%%%%%%%%%%%%%%%%%%%%%%%%%%%%
In Section~\ref{subsec:SLS}, our next result is a different approach to compute the level value for an interval.   This approach avoids the lower bound for a direct implementation by maintaining additional information associated with the intervals that have been colored.  
% beeOur approach avoids the computation of the maximum clique size in the interval graph. Thus, we avoid a direct implementation of KT-algorithm in computing the level value for an interval. 
 While our approach, called Algorithm $\KTSLS$, uses the same number of colors as the KT-algorithm, we show that the level value for each interval computed by our approach is at most the level value computed by the KT-algorithm
( Lemma~\ref{lem1: max-height-equal-p(v)}).  We show an example where for an interval the level value computed by Algorithm $\KTSLS$  is smaller than the level value computed by the KT-algorithm.  We design an incremental interval coloring algorithm which implements Algorithm $\KTSLS$ in Section~\ref{sec: Incremental-Algorithm} and show that it uses at most $3 \omega -2$ colors.
\begin{itemize}
	\item There exists an incremental algorithm which supports insertion of an interval in amortized $O(\log n + \Delta)$ time per update, where $n$ is the total number of intervals in the update sequence and $\Delta$ is the maximum degree of a vertex in the interval graph formed by those intervals. (Theorem~\ref{thm:IncrementalAlgo})
\end{itemize}
%%%%%%%%%%%%%%%%%%%%%%%%%%%%%%
In Section~\ref{sec:fully-dynamic-algorithm}, in the fully dynamic setting, an interval that has already been colored can be deleted, apart from the insertions. At the end of each update, our aim is to maintain a $3 \omega - 2$ coloring of the remaining set of intervals, where $\omega$ is the maximum clique in the interval graph associated with the remaining set of intervals. In order to bound the number of colors to $3 \omega - 2$, deletion of an interval may trigger a change in the colors of some of the remaining intervals creating a set of \textit{dirty} intervals. Cleaning up of those dirty intervals may in turn create more dirty intervals resulting in a cascading effect. We design an approach to efficiently compute the set of such dirty intervals after a deletion. Thus, we present a fully dynamic algorithm for $3 \omega - 2$ interval coloring in the fully dynamic setting.
\begin{itemize}
	\item There exists a fully dynamic algorithm which supports insertion of an interval in $O(\log n + \Delta \log \omega)$ update time and deletion of an interval in $O(\Delta^2 \log n)$ update time in the worst case, where $n$ is the total number of intervals inserted and $\Delta$ is the maximum degree of a vertex in the interval graph formed by those intervals. (Theorem~\ref{thm : update-time-fully-dynamic})   
\end{itemize}
%%%%%%%%%%%%%%%%%%%%%%%%%%%%%%%%%%%%%%%%%%%
Our final contribution is motivated by the fact that the KT-algorithm computes the maximum clique size in an induced subgraph of the neighbors of the current interval.  In our attempt to design efficient data structures to report the neighborhood of a vertex we encountered a connection to the online boolean matrix vector multiplication problem and the related OMv conjecture, which is due to Henzinger et. al \cite{DBLP:conf/stoc/HenzingerKNS15}. 
We present a reduction in Section~\ref{sec:OMv} where we show the following result.
\begin{itemize}
\item Any algorithm that needs to compute induced subgraph among the neighbors of a given vertex requires at least quadratic time unless online boolean matrix vector multiplication conjecture is false.(Theorem~\ref{thm:InducedSubGraphviaOMv})
\end{itemize}	  
Finally, we use the well-known interval tree data structure to obtain the following result on online boolean matrix vector multiplication conjecture. 
\begin{itemize}
\item In the online boolean matrix vector multiplication problem, if the boolean matrix and the vectors in the online sequence have consecutive ones property then the OMv conjecture is false. (Theorem~\ref{thm:OMvConjConsOnes})
\end{itemize}

 %Finally, the question of significant interest to us is whether the dependence on $\Delta$ can be sub-linear in the incremental case and whether it can be sub-quadratic in the fully dynamic case.   Another interesting direction is the nature of the trade-off between the number of colors used and the update time if we allow  a change of color assigned to an interval.

%% file: Journal/preliminaries.tex
% !TEX root = ../IntervalColoringFullVersion.tex
\section{Kierstead-Trotter algorithm and Supporting Line Segment} 
\label{sec:prelims}
Let the set $\mathcal{I} = \{I_1, I_2, \dots, I_n\}$ denote a sequence of $n$ intervals, and let $G(\mathcal{I})$ denote the associated interval graph.   For $1 \leq j \leq n$, let $I_j = [l_j, r_j]$ where  $l_j$ and  $r_j$ represent the left and right endpoint of $I_j$, respectively. Let $\sigma = v_1,v_2,v_3,\dots,v_n$ be the ordering of vertices of interval graph $G = G(\mathcal{I})$ where vertex $v_j$ is the $j$-th vertex in $\sigma$ and it corresponds to the interval $I_j$ in $\mathcal{I}$. Let $\omega(G)$, $\Delta(G)$, and $\chi(G)$ denote the size of the maximum cardinality clique in G, the maximum degree of a vertex in $G$, and  the chromatic number of $G$, respectively. It is well-known that for interval graphs $\omega(G)=\chi(G)$. 
When the graph $G$ is clear, refer to these numbers as $\omega$,  $\Delta$, and $\chi$. 
\subsection{Kierstead-Trotter algorithm - overview}
\label{KTO}
The intervals in the sequence $\mathcal{I}$ are presented to the online KT-algorithm.  For $i \geq 0$, let $I_i$ be the interval presented and let $v_i$ be the corresponding vertex in  $\sigma$.  The KT-algorithm computes a color  based on the color given to the vertices $v_1, \ldots, v_{i-1}$.  The color assigned to a vertex $v$ is a tuple of two values and is denoted as $(p(v),o(v))$. $p(v)$ is called the level value, $o(v) \in \{1,2,3\}$ is called the offset, and $v$ is said to be in level $p(v)$.  $p(v)$ is computed in Step I and in Step II $o(v)$ is computed.  The key property is that for each edge $\{u,v\}$, the tuple $(p(u), o(u))$ is different from $(p(v), o(v))$. \\
\noindent
\textbf{Step I: } 
For $r \geq 0$, let $G_r(v_i)$ denote the induced subgraph of $G$ on the vertex set  $\{v_j | v_j \in V(G), j < i, p(v_j) \leq r,(v_i,v_j) \in E(G)\}$.  Define 
$p(v_i)$ = $\min\{r | \omega(G_r(v_i)) \leq r\}$.\\\\\\
 
\noindent
\textbf{Key Properties maintained by Step I \cite{kierstead1981extremal}: }
\begin{itemize}
	\item For each vertex $v_i$, $p(v_i)\leq \omega-1$.
	\item {\bf Property P :} The set $\{v | p(v) = 0\}$ is an independent set.  For each $i$, $1\leq l \leq \omega-1$,  the subgraph of $G$ induced on $\{v \mid p(v) = l\}$ has maximum degree at most 2.  
\end{itemize}
\textbf{Step II: }
%Since there are at most two neighbours of $v_i$ such that their level is $p(v_i)$,   
$o(v_i)$ is chosen to be the smallest value from the set $\{1,2,3\}$  which is different from the offset  of each of the at most two neighbors whose level is $p(v_i)$.  \\  
\noindent
{\bf Analysis:}
Since the vertices with level value $0$ form an independent set, the offset for all these vertices is $1$. Therefore, the color for all the vertices in level $0$ is $(0,1)$. By Property P, for each level $1 \leq l \leq \omega-1$,  the maximum degree in the graph induced by vertices in the level $l$ is $2$. Therefore, the algorithm uses at most $3$ colors, $(l,1)$, $(l,2)$, and $(l,3)$, to color the vertices in level $l$. Hence, total number of colors used by the algorithm is at most $1 + 3(\omega - 1)$ = $3\omega - 2$.
%%%%%%%%%%%%%%%%%%%%%%%%%%%%%%%%%%%%%%%%%%%%%%%%%%%%%%%%%%%%%%%%%%%%%%%%%
\subsection{Quadratic lower bound for computing the $\omega$ in the graph induced by a subset of neighbors of a vertex}
\label{sec:KTalgorithm-hardness}
We start by considering  implementations of the KT-algorithm in which the data structures are designed only to store the input intervals and support only intersection queries among intervals.  
We refer to such an implementation as a {\em direct implementation} and prove a lower bound on the running time of a direct implementation.  This lower bound motivates the additional data structures that are necessary to obtain an implementation  of the KT-algorithm with a better running time.
We start by observing a lower bound on the time to identify the size of a maximum clique in a given set of intervals.  
%%%%%%%%%%%%%%%%%%%%%%%%%%%%%%%%%%%%%%%%%%
\begin{lemma}
\label{lem:LowerBound-maximumClique}
A deterministic algorithm which computes the maximum clique size in the interval graph formed by a given set of $n$ intervals has running time $\Omega(n)$.
\end{lemma}
\begin{proof}
The proof is by contradiction.
Let $\mathcal{B}$ be a deterministic algorithm such that on each  input consisting of a set of $n$ intervals,  it reports the maximum clique size in the corresponding interval graph  in $o(n)$ time. Due to this assumption, it follows that algorithm $\mathcal{B}$ does read the entire input on $n$ intervals.  Let us consider the execution of $\mathcal{B}$ on interval sequence $\mathcal{I}_1=\{[0,x_1], [0, x_2], \dots, [0, x_j], \dots,[0,x_n]\}$.  Let $1 \leq j \leq n$ be the index such that during the execution, $\mathcal{B}$ does not read the $j$-th interval.  Consider $\mathcal{I}_2$  obtained from $\mathcal{I}_1$  by replacing the $j$-th interval by an interval disjoint from all the other $n-1$ intervals.  Since the execution of $\mathcal{B}$ does not read the $j$-th interval, the output on both $\mathcal{I}_1$ and $\mathcal{I}_2$ 
will both be the same.  However, both the outputs cannot be correct since the maximum clique size for $\mathcal{I}_1$ is $n$ and for $\mathcal{I}_2$ is $n-1$.  Therefore, for a deterministic algorithm $\mathcal{B}$ to compute the maximum clique size on all inputs, it must read all the intervals in the input, and thus its running time is 
$\Omega(n)$. Hence the Lemma.
%
%Let us consider the two inputs  {\sf input1} = $\{[0,x_1], [0, x_2], \dots, [0, x_j], \dots,[0,x_n]\}$ and {\sf input2} = $\{[0,x_1], [0, x_2], \dots, [x_j, x_j], \dots,[0,x_n]\}$  where  for $1 \leq i \leq n$, $x_i$ is positive  for $1 \leq i \leq n$, $x_i$ is positive  and $x_j$ is not the minimum in $\{x_1, x_2, \dots, x_j,\dots, x_n\}$.   Now, if the same execution of $\mathcal{B}$ is performed on {\sf input1} and {\sf input2}, the output in both the cases would have been the same, because for an algorithm that does not read the $j$-th interval in the input, the remaining $n-1$ intervals in the input are the same.    
%Observe the execution of $\mathcal{B}$ on {\sf input1} = $\{[0,x_1], [0, x_2], \dots, [0, x_j], \dots,[0,x_n]\}$ and  where, For {\sf input1} the correct output is $n$ and for {\sf input2} the correct output is $n-1$. But $\mathcal{B}$ does not read the $j$-th interval in the input, that is interval $[0, x_j]$ in {\sf input1} and interval $[x_j, x_j]$ in {\sf input2}. As a consequence, the output of $\mathcal{B}$ for both {\sf input1} and {\sf input2} is same. %$n-1$. 
%This implies that if $\mathcal{B}$ is allowed $o(n)$ time for execution then $\mathcal{B}$ cannot distinguish between {\sf input1} and {\sf input2}. Therefore, our assumption that $\mathcal{B}$ outputs the correct result for any input of length $n$ in $o(n)$ time is incorrect. Further, $\mathcal{B}$ must read the entire input to distinguish between {\sf input1} and {\sf input2}. Hence, $\mathcal{B}$ takes $\Omega(n)$ time.           
\end{proof}

\noindent
Using Lemma~\ref{lem:LowerBound-maximumClique} we prove that a direct implementation of the KT-algorithm will have a $\Omega(\Delta^2)$ running time.
%
%In this section we show that we are unlikely to obtain a sub-quadratic time implementation of the KT-algorithm.\\ 
%{\bf Crucial step in the implementation of KT-algorithm:} 
%For a given vertex $v_i$, when the KT-algorithm is presented with the corresponding interval $I_i$, it computes the set of intervals intersecting with $I_i$ which have level value at most $r$. Then, it computes the size of the maximum clique in the interval graph formed by those intervals. KT-algorithm needs to repeat this computation for different values of $r$ starting from $0$ until obtaining the minimum value of $r$ for which the condition $\omega(G_r(v_i)) \leq r$ is true.\par   
%We first show that it is unlikely to obtain a sub-linear time algorithm which computes the cardinality of the maximum clique in the interval graph formed by a given set of intervals. 
%%%%%%%%%%%%%%%%%%%%%%%%%%%%%%%%%%%%%%%%%%%%%%%%%%%%%%%%%%%%%%%%%%%%%%%%%%%%%%%%%%%%%%%%%%%%%%%%%%%%%
\begin{theorem}
	\label{thm:LowerBound-KT-algorithm}
	A direct implementation of  the KT-algorithm has $\Omega(\Delta^2)$ running time where $\Delta$ is the maximum degree of a vertex in the associated interval graph.
\end{theorem}
\begin{proof}
	%For a given interval $I_i$, KT-algorithm computes the color $(p(v_i), o(v_i))$ based on the adjacency information and the color given to the intervals $I_1, I_2, \dots, I_{i-1}$.  
	For each $i \geq 1$ and $r \geq 0$, to check if $p(v_i) = r$, the KT-algorithm computes the maximum clique size in the interval graph formed by the intervals with level value at most $r$ and intersecting with the input interval $I_i$. From Lemma~\ref{lem:LowerBound-maximumClique}, computing the size of the maximum clique takes $\Omega(\Delta)$ time.  The KT-algorithm repeats this computation for the values of $r$ starting from $0$ until the value for which $\omega(G_r(v_i)) \leq r$ is true. The worst case is reached for the input sequence $\mathcal{I}_1$  in Lemma~\ref{lem:LowerBound-maximumClique} for which the clique will be computed in the graphs $G_0(v_i), G_1(v_i), \ldots, G_{i}(v_i)$.  The running time of a direct implementation on $\mathcal{I}_1$ is $\Omega(\omega^2)$ which is $\Omega(\Delta^2)$.  Therefore, a direct implementation of the KT-algorithm takes $\Omega(\Delta^{2})$ time in the worst case. Hence the Theorem.
\end{proof}
%%%%%%%%%%%%%%%%%%%%%%%%%%%%%%%%%%%%%%%%%%%%%%
%We strengthen our result by showing that any data structure which maintains a set of intervals, supports insertion of a new interval, computes the number of intervals intersecting with a given interval (intersection query), and uses comparison operation among the previously inserted intervals before inserting a new interval and responding to an intersection query requires $\Omega(\log n)$ comparisons either during insertion or during query.
\noindent
While our subsequent results show that we can maintain additional information to circumvent the lower bound faced by a direct implementation, we observe that any comparison based data structure to maintain the given set of intervals and respond to intersection queries uses $\Omega(\log n)$ comparisons.
%%%%%%%%%%%%%%%%%%%%%%%%%%%%%%%%%%%%
%\begin{lemma}
%\label{lem:Lowerbound-DS-Interval}
%Any data structure that uses comparison operation among the previously inserted intervals before inserting a new interval and responding to an intersection query takes $\Omega(\log n)$ comparisons either during insertion or during query, where $n$ is the number of intervals already inserted in the data structure.
%\end{lemma}
\begin{lemma}
\label{lem:Lowerbound-DS-Interval}
Let {\sf D} be a comparison based data structure which supports the following operations on an interval $[l,r]$:
	\begin{itemize}
		\item {\sf D.Insert([l, r])}: Inserts interval {\sf [l, r]} into {\sf D}.  
		\item {\sf D.Query([l, r])}: Returns the total number of intervals in {\sf D} which are intersecting with {\sf [l, r]}. 
	\end{itemize}
	Let $c_u$ be the number of comparisons performed during {\sf D.Insert([l, r])} before inserting [l, r]. Let $c_q$ be the number of comparisons performed during {\sf D.Query([l, r])} before responding to the query.  Then either $c_u$ or $c_q$ is $\Omega(\log n)$.
\end{lemma}
\begin{proof}
	It is possible to use data structure {\sf D} to design a comparison based sorting algorithm which is defined as follows: Let $x_1, x_2, \dots, x_n$ be $n$ distinct numbers given as input for comparison sorting. For every $1 \leq i \leq n$, perform ${\sf D.Insert([x_i, x_i])}$.
	%data structure {\sf D} is used to perform an insert operation, ${\sf D.Insert([x_i, x_i])}$.   
	A linear search is performed on $x_1, x_2, \dots, x_n$ to find the minimum denoted by {\sf min}. Finally, to compute the sorted order, for each $1 \leq i \leq n$, perform ${\sf D.Query([min, x_i])}$.
	% For every $i$ in $[1, n]$, data structure {\sf D} is used to perform a query, ${\sf D.Query([min, x_i])}$.
%	If the query ${\sf D.Query([min, x_i])}$ returns $1$ then $x_i$ is the first element in the sorted order. In general, 
If  query ${\sf D.Query([min, x_i])}$ returns $j$, then $x_i$ is the $j$-{th} element in the sorted order.  Thus, the total number of comparisons required to find the sorted order is $n \cdot c_u$ + $n$ + $n \cdot c_q$.  It is well-known that any comparison sorting of $n$ numbers performs $\Omega(n \log n)$ comparisons \cite{Cormen:LB-Sorting}. Therefore, $n \cdot c_u$ + $n$ + $n \cdot c_q$  is $\Omega (n \log n)$. This implies that either $c_u$ = $\Omega(\log n)$ or $c_q = \Omega(\log n)$. Hence the Lemma.	
\end{proof}
%%%%%%%%%%%%%%%%%%%%%%%%%%%%%%%%%%%%%%%
%%%%%%%%%%%%%%%%%%%%%%%%%%%%%%%%%%%%%%%%%%%%%%%%%%%%%%%%%%%%%
\noindent
The result in Lemma~\ref{lem:Lowerbound-DS-Interval} shows that any interval tree based approach which maintains the input intervals and computes intersecting intervals will use $\Omega(\log n)$ comparisons.  
% if interval tree like data structure is used to maintain the set of intervals.
On the other hand, in Section ~\ref{subsec:SLS} we overcome the  lower bound presented in Theorem~\ref{thm:LowerBound-KT-algorithm} by maintaining additional information about the coloring computed by the KT-algorithm. This additional information plays a crucial role in an efficient data structure for the KT-algorithm.
%deters us from coming up with a direct implementation of the KT algorithm to obtain an incremental interval coloring algorithm. 
%In Section~\ref{subsec:SLS}, we present an alternate approach. 
%%%%%%%%%%%%%%%%%%%%%%%%%%%%%%%%%%%%%%%%%%%%%%%%%%%%%%%%%%%%%%%%%%%%%%%%%%   
%%%%%%%%%%%%%%%%%%%%%%%%%%%%
\subsection{Supporting Line Segment (SLS)-a geometric handle}
\label{subsec:SLS}
To overcome the limitation of computing a maximum clique in an induced subgraph, we maintain the size of some cliques, and use the structure of interval graphs to conclude that these cliques indeed represent a  maximum clique in the neighborhood of each interval.  The algorithm can be seen as an efficient version of the KT-algorithm and we refer to it as $\KTSLS$ (KT-algorithm using supporting line segment). 
%In this section we present an alternate approach to compute the level value for an interval which is at most the level value computed by the KT-algorithm. 
%Let $\mathcal{W}$ = $\{0,1,\dots,\omega-1\}$.   

For each $i \geq 1$,  $L(I_i) \in \{0,1,\dots,\omega-1\}$  and $p(v_i)$ denote the level value computed by $\KTSLS$ and the KT-algorithm, respectively.
Further, for each $i \geq 1$, $o(I_i)$ and $o(v_i)$ denote the offset computed by $\KTSLS$ and the KT-algorithm, respectively.
Let $t$ be a non-negative real number and $\mathcal{I}_t$ be the set of all intervals in $\mathcal{I}$ which contain the point $t$. For the set $\mathcal{I}_t$, define the set $levels(\mathcal{I}_t) = \displaystyle \bigcup_{I \in \mathcal{I}_t} \{L(I)\}$
%= $\{L(I_j) \in \{0,1,\dots,\omega-1\} | I_j \in \mathcal{I}_t \}$ 
to be the set of levels assigned to intervals in $\mathcal{I}_t$. Define $h_t$ = $\min(\{y \in \{0,1,\dots,\omega-1\} | y \notin levels(\mathcal{I}_t) \})$. In other words, $h_t$ is the smallest non-negative integer which is not the level value for an interval containing $t$.   For $h_t \geq 1$, the \textit{Supporting Line Segment}(SLS) at $t$ is defined to be the  
set $e_t$ = $\{(t,0) \dots (t,h_t-1)\}$, and  $h_t$ is called the \textit{height} of the SLS $e_t$.  Note that the set $\mathcal{I}_t$ is of size at least $h_t$, and there are $h_t$ intervals in $\mathcal{I}_t$ for which the level values are $0$ to $h_t-1$.  

\noindent
%\begin{definition}
%\label{def:SLS}	
%For $h_t \geq 1$, \textit{Supporting Line Segment}(SLS) at $t$ is defined to be the  
%set $e_t$ = $\{(t,0) \dots (t,h_t-1)\}$, and  $h_t$ is called the \textit{height} of the SLS $e_t$.
%\end{definition}
%\noindent
%%%%%%%%%%%%%%%%%%%%%%%%%%%%%%%%%%%%%%%%%%%%%%%%%%%%%%%%%%%%%%%%%%%%%%%%%%%%%%%%%%%%%%%%%%%%%%%%%%%%%%%%%%%%%%%%%%%%
\textbf{Algorithm $\KTSLS$:} For $i \geq 1$, the color $(L(I_i), o(I_i))$ for $I_i$ is computed as follows:
%For the $i$-th interval $I_i$ presented to $\KTSLS$, it computes the color $(L(I_i), o(I_i))$ as follows:
%%%%%%%%%%%%%%%%%%%%%%%%%%%%%%%%%%%%%%%%%
\begin{enumerate}
	\item {\bf Step I:} For each $t \in I_i$,  compute the height, $h_t$,  of the SLS $e_t$. Define $L(I_i) = \displaystyle \max_{t \in I_i} h_t$.\\
	%%%%%%%%%%%%%%%%%%%%%%%%%%%%%%%%%%
	\textbf{Key Properties maintained by Step I} 
	\begin{itemize}
		\item For each interval $I_i$, $L(I_i) \leq p(v_i)\leq \omega-1$(Lemma~\ref{lem1: max-height-equal-p(v)}).
		\item {\bf Property P :} The set $\{I | L(I) = 0\}$ is an independent set.  For each $i$, $1\leq i \leq \omega-1$,  the subgraph of $G$ induced on $\{I \mid L(I) = i\}$ has maximum degree at most 2 (Lemma~\ref{lem: Property-P}).  
	\end{itemize}
	 
	%%%%%%%%%%%%%%%%%%%%%%%%%%%%%%%%%%%%%%%%%
	\item {\bf Step II:} :
	Compute $o(I_i)$ to be the smallest value from the set $\{1,2,3\}$ that is different from the offset of the neighbours of $I_i$ which have the level $L(I_i)$. 
\end{enumerate}
%%%%%%%%%%%%%%%%%%%%%%%%%%%%%%%%%%%%%%%%%%%
\begin{figure}[htbp]	
	\flushleft	
	\includegraphics[scale=1]{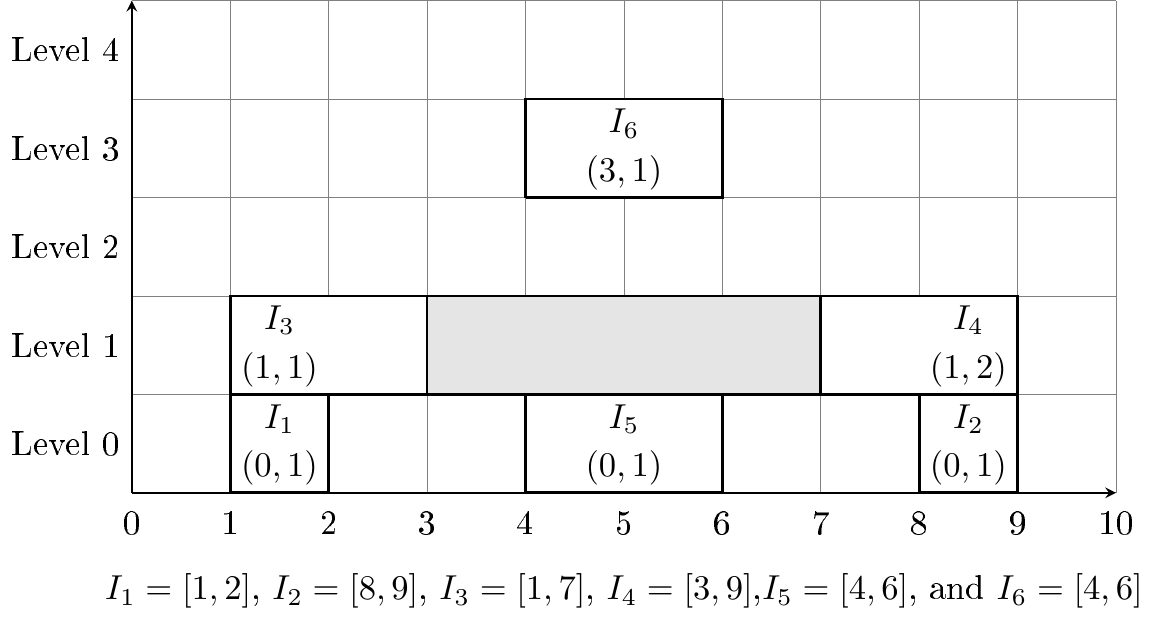}
	\caption{For $1 \leq i \leq 6$, color $(p(v_i), o(v_i))$ is shown along with the interval $I_i$. For $1 \leq i \leq 5$, $(p(v_i), o(v_i))$ = $(L(I_i), o(I_i))$.\\$~$\hspace{1.2cm} For $I_6$, $(L(I_6), o(I_6))$ = $(2,1)$, whereas $((p(v_6), o(v_6))$ = $(3,1)$.}
	\label{fig:example}	
\end{figure}
{\bf Remark:}
We show an example in Figure~\ref{fig:example}, where for interval $I_6$, the level value computed by $\KTSLS$  is strictly less than the level value computed by the KT-algorithm.
The intervals arrive in the following order : $I_1 = [1,2], I_2 = [8,9], I_3 = [1,7], I_4 = [3,9], I_5 = [4,6]$, and $I_6 = [4,6]$. 
Intervals $I_1$, $I_2$, and $I_5$ get level value $0$. KT-algorithm computes $p(v_6) = 3$ and $\KTSLS$ computes $L(I_6) = 2$. The portion colored as gray is the overlapping portion between interval $I_3$ and interval $I_4$.\\
%%%%%%%%%%%%%%%%%%%%%%%%%%%%%%%%%%%%%%%%%%%%%%%%%%%%%%%%%%%%%%%%%%%%%%%%%%%%%%%%%%%%%%%%
%%%%%%%%%%%%%%%%%%%%%%%%%%%%%%%%%%%%%%%%%%%%%%%%%%%%%%%%%%%%%%%%%%%%%%%%%%%%%%%%%%%
{\bf Correctness of Algorithm $\KTSLS$.}  We start by proving that the level value computed by Algorithm $\KTSLS$ depends only on the set of endpoints of the intervals.  This  finite set is denoted by $\mathcal{E}$ and we refer to the points in the set as endpoints.  
%The correctness follows from Lemma~\ref{lem:edge-point-in-all-intervals}, Lemma~\ref{lem1: max-height-equal-p(v)}, 
%Lemma~\ref{lem2: max-height-equal-p(v)}, 
%and Lemma~\ref{lem: Property-P}.\\ 
%%%%%%%%%%%%%%%%%%%%%%%%%%%%%%%%%%%%%%%%%%%
%We prove using Lemma~\ref{lem:edge-point-in-all-intervals}  
%that the $L(I_i)$ can be computed based on the height of the SLS at finite set of points in interval $I_i$. 
%In particular, we show that the finite set of points that we need to consider in interval $I_i$ is the set of endpoints of the intervals intersecting with $I_i$.
%%%%%%%%%%%%%%%%%%%%%%%%%%%%%%%%%%%%%%%%%%%%%%%%
\begin{lemma}
	\label{lem:edge-point-in-all-intervals}
	%Let $t$ be a non-negative real number and let $\mathcal{I}_t$ be the set of  intervals that contain $t$. Then there exists at least one endpoint in the set $\mathcal{E}$ which is  contained in  each interval in $\mathcal{I}_t$. Further, the height of the SLS at this endpoint is at least the height of the SLS at $t$.
	For each real number $t$, $\mathcal{E} \cap \displaystyle \bigcap_{I \in \mathcal{I}_t} I \neq \emptyset$.  Further, there is an endpoint $p \in \mathcal{E}$ such that $h_p \geq h_t$. 
\end{lemma}
\begin{proof}
	If $t$ is an endpoint of an interval, then $t \in \mathcal{E}$ and hence the Lemma is proved. Suppose $t$ is not an endpoint. Let $l_{t}$ denote the largest left endpoint among all the intervals in $\mathcal{I}_t$ and $r_{t}$ denote the smallest right endpoint among all the intervals in $\mathcal{I}_t$. By definition, $l_{t} \in \mathcal{E}$ and $r_{t} \in \mathcal{E}$. Since $\mathcal{I}_t$ is a set of intervals, it follows that $l_{t}$ and $r_{t}$ are present in all the intervals in $\mathcal{I}_t$. Further, since both $l_t$ and $r_t$ are present in each interval in $\mathcal{I}_t$, it follows that the set of intervals that contain them is a superset of $\mathcal{I}_t$.  Therefore, the height of the SLS at $l_t$ and $r_t$ is at least the height of the SLS at $t$.  Hence the Lemma.   
\end{proof}
\noindent
%%%%%%%%%%%%%%%%%%%%%%%%%%%%%
From the description in Section \ref{KTO}, we know that the level value of $I_i$ computed by the KT-algorithm is given by $p(v_i)$ = $\min\{r|\omega(G_{r}(v_{i})) \leq r\}$. Further in Algorithm $\KTSLS$, $L(I_i)$ is defined to be  the maximum height of the SLS at the endpoints contained in $I_i$.  We next prove that $L(I_i) \leq p(v_i) \leq \omega -1$.
%%%%%%%%%%%%%%%%%%%%%%%%%%%%%%%%%%%%%%%%%
\begin{lemma}
	\label{lem1: max-height-equal-p(v)}
	For each $i \geq 1$, $p(v_i)$  is at least the maximum  height of the SLS at any endpoint contained in the interval $I_i$.  
\end{lemma}
\begin{proof}
	%%%%%%%%%%%%%%%%%%%%%%%%%%%%%%%%%%%%%%%%%%%%%
	By definition, $L(I_i)$ is the maximum height of SLS at any endpoint contained in $I_i$.  By the definition of the height of an SLS at an endpoint $t$, we know that for each $0 \leq r \leq h_t-1$ there is an interval $I \in \mathcal{I}_t$ such that $L(I) = r$, and all these intervals form a clique of size $h_t$.  Therefore, it follows that $G_{L(I_i)}(v_i)$ has a clique of size at least $L(I_i)$.  Therefore, it follows that $p(v_i) \geq L(I_i)$. Hence the Lemma.
	%%%%%%%%%%%%%%%%%%%%%%%%%%%%%%%%%%%%%%%%%%%%%%% 
\end{proof}
\noindent
%%%%%%%%%%%%%%%%%%%%%%%%%%%%%%%%%%%%%%%%%%%%%%%%%%%%%%%%%%%%%%%%%%%%%%%%%%%%%%%%%%%
We prove in Lemma~\ref{lem: Property-P} that the level values computed by  Algorithm $\KTSLS$ satisfy Property \textbf{P}.\\\\\\
\begin{lemma}
	\label{lem: Property-P}
	 Algorithm $\KTSLS$ satisfies Property \textbf{P} and thus uses at most $3 \omega - 2$ colors.
\end{lemma}
\begin{proof}
	We first prove that the set $\{I | L(I) = 0\}$ is an independent set. To show this, we prove that for a pair of intersecting intervals $I$ and $J$, at least one of $L(I)$ or $L(J)$ is more than 0.  
	Without loss of generality, let us assume that the interval $I$ appeared before $J$. If $L(I) > 0$, then our claim is correct.  We now consider the case when $L(I) = 0$.  Since $I$ and $J$ intersect, it follows that an endpoint of one of them is contained in the other.  Therefore, after $J$ is presented to Algorithm $\KTSLS$, the SLS of one of the endpoints in $I \cup J$ is more than 0.  
	%Therefore, at the time when $I_j$ appears,  there is an endpoint of $I_i$ contained in $I_j$ where the height of the SLS is non-zero.  
	By the definition of $L(J)$ in Algorithm $\KTSLS$,  it follows that $L(J) > 0$.  Therefore, $\{I | L(I) = 0\}$ is an independent set.  The same argument also shows that if $L(I) = L(J)$, then $I \not\subseteq J$ and $J \not\subseteq I$.
	
%	Further, for each $1 \leq l \leq \omega-1$, the same argument is used to show that  for each pair of intervals $I_i$ and $I_k$ satisfying $L(I_i) = L(I_k) = l$, $I_i \not\subseteq I_k$. Recall that vertex $v_j$ corresponds to interval $I_j$ in the associated interval graph.	
We now prove that the level value computed by Algorithm $\KTSLS$ satisfies Property \textbf{P}.  Let $I$ be the first interval during the execution of the algorithm which has at least 3 intersecting intervals in level $l=L(I)$.  Let these intervals be $J_1, J_2, J_3$.  Since there cannot be a containment relationship between two intersecting intervals with the same level value, it follows that two intervals contain a common endpoint with $L(I)$.  Without loss of generality, let $J_1$ and $J_2$ contain a common endpoint of $I$.  Further, we know by the Helly property for intervals that one of the intervals is contained in the union of the other two \cite{golumbic1980algorithmic}.  
%	We prove that for each $l$, $1 \leq l \leq \omega-1$, the subgraph of $G$ induced on $\{v_j| L(I_j) = l\}$ has maximum degree at most $2$. Suppose not, and if some vertex is of degree 3 in level $l$.   We know that for two vertices with the same level value, the corresponding intervals cannot have a containment relationship between them.  Therefore, it follows that the vertex of  degree 3 is in a clique of 3 vertices in level $l$.   Let $v_i$, $v_j$, and $v_k$ constitute the clique of 3 vertices in level $l$.  The  intervals corresponding to the three vertices are such that one of the intervals is contained in the union of the other two.
  Consequently, one of the 3 intervals contains a point $t$ for which the SLS has height $l+1$.  This contradicts the hypothesis that $L(I) = l$.  Therefore, our assumption that an interval $I$ has at least 3 neighbours in $L(I)$ is wrong.  Consequently, all the intervals with the same level value are assigned an offset from $\{1,2,3\}$. Further, if two intervals with the same level value intersect, then they get different offsets as described in {\bf Step II}.  
%  3 intervals cannot be assigned the same level value.  Therefore, for each level $l$, the maximum vertex degree in the graph induced on $\{v_j| L(I_j) = l\}$ is at most 2. \par
%	Since we have proved that the level value computed by $\KTSLS$ satisfies Property \textbf{P}, we  use the same procedure as KT-algorithm (described in Section~\ref{KTO}) to compute the offset. 
	From Lemma \ref{lem1: max-height-equal-p(v)}, we know that
	for any interval $I_i$, we have $L(I_i) \leq p(v_i)$. Therefore, maximum level value of any interval is $\omega-1$. For level $0$ we use one color and for every other level we use at most $3$ colors. Therefore, the number of colors used by ALgorithm $\KTSLS$ is $3(\omega - 1) + 1$ = $3\omega - 2$. Hence the Lemma.
\end{proof}
\noindent
In the rest of the paper we design data structures that are useful in an efficient implementation of Algorithm $\KTSLS$ in both the incremental and fully dynamic settings.  Apart from data structures to maintain intervals, we also use data structures to maintain supporting line segments.  The data structures to maintain the supporting line segments are crucial in overcoming the limitations of a direct implementation of the KT-algorithm.  The necessary data structures are described in Section \ref{subsec:DSpreliminaries}.  
%If we obtain an efficient implementation of the $\KTSLS$ algorithm then it gives us an algorithm for interval coloring in the dynamic setting. In Section~\ref{subsec:DSpreliminaries}, we describe the data structures we use in the incremental and fully dynamic setting to maintain the set of intervals, set of endpoints and supporting line segments.
\noindent
\subsection{Dynamic Data Structures  for Algorithm $\KTSLS$}
\label{subsec:DSpreliminaries}
%%%%%%%%%%%%%%%%%%%%%%%%%%%%%%%%%%%%%%%%%%%%%%%%%%%%%
\begin{table}[h]
	\centering
	\begin{tabular}{|l|l|l|} 
		%		\hline
		%		\multicolumn{3}{|l|}{}\\
		\hline
		Procedure & Incremental  & Fully Dynamic\\
		%& Running Time &  \\
		\hline
		$\getSLS$$(\mathcal{I},t)$:Maintains the SLS  for endpoint $t$  & Worst case & Worst case \\
		(i) in the Incremental case:  as a dynamic array $A_t$  & $O(\log(n)+\omega)$& $O(\log(n)+\omega \log \omega)$  \\
		and doubly linked list $Q_t$. & (Lemma~\ref{lem:runningtime-computeSLS})& (Lemma~\ref{lem:runtimeComputeSLSfullyDynamic})\\
		(ii) in the Fully Dynamic case: Red Black Tree $Z_t$ and ${NZ}_t$.& Return value $A_t, Q_t$  &  Return value $Z_t, {NZ}_t$\\
		\hline
		%		\multicolumn{3}{|l|}{}\\
		%		\hline
		$\computemaxheightSLS$$(\mathcal{E},I_i)$: From the interval tree $\mathcal{E}$, &Worst case  &Worst case \\
		computes the set of endpoints $S$  contained in the interval $I_i$  & $O(\log(n)+\Delta)$ & $O(\log(n)+\Delta \log \omega)$  \\
		and returns $h$ = $\max\{h_t | t \in S\}$  &  (Lemma~\ref{lem:runningTimeMaxHeightofSLS}) & (Lemma~\ref{lem:runtimeComputemaxheightfullydynamic}) \\
		&Return value $S,h$  & Return value $S,h$\\
		\hline
		$\updateEndPoints$$(S,L(I_i))$: Updates the endpoints in $S$& Amortized  & Worst case\\
		on addition of  interval $I_i$ at level $L(I_i)$. For each $t \in S$ & $O(\Delta)$ & $O(\Delta \log \omega)$ \\
		(i) Incremental case: updates $A_t$  and $Q_t$ &(Lemma~\ref{lem:runningTimeUpdateEndpoints})   & (Lemma~\ref{lem:runtimeUpdaeSLSfullydynamic})   \\
		(ii) Fully Dynamic case: updates $Z_t$ and ${NZ}_t$ & & \\
		\hline
		$\OFFSET$$(I_i)$: Assigns an offset value to $I_i$ from $\{1,2,3\}$& Worst case & Worst case \\
		by considering the offset of the intervals  &  $O(\log(n))$ & $O(\log(n))$\\
		intersecting it in $T[L(I_i)]$ &(Lemma~\ref{lem:time-greeyd-coloring}) & (Lemma~\ref{lem:time-greeyd-coloring})\\
		\hline
	\end{tabular}
	\caption{Comparison of Procedures in Incremental and Fully Dynamic cases}
	\label{table:2}
\end{table}
%%%%%%%%%%%%%%%%%%%%%%%%%%%%%%%%%%%%%%%%%
\noindent
In this section, we present the various data structures used to implement $\KTSLS$ in the incremental and fully dynamic setting. For this purpose, we come up with the procedures listed in Table~\ref{table:2}. The procedures in the incremental setting differ from their counterparts in the fully dynamic setting on the data structures used to store SLS. The running time of the incremental and the fully dynamic algorithms are governed by the running times of these procedures.
Detailed descriptions of the procedures are given in Section~\ref{subssec:procedures-handle-insert} for incremental setting and in Section~\ref{subsec:procedureHandleDeleteInsert} for fully dynamic setting. Next, we describe the different data structures which are used to implement the procedures in Table~\ref{table:2}. Running time of different operations on these data structures are listed in Table~\ref{table:1}. 
\noindent   
%%%%%%%%%%%%%%%%%%%%%%%%%%%%%%%%%%%%%%%%%%%%%%%%%%%%%
%%%%%%%%%%%%%%%%%%%%%%%%%%%%%%%%%%%%%%%%%%%%%%%%%%%%%
\begin{table}[htb]
	\centering
	\begin{tabular}{|l|l|c|c|} 
		\hline
		\multicolumn{4}{|l|}{Interval Tree $\mathcal{I}$ \cite{de2008computational}}\\
		\hline
		Method & Description & Running Time & Return \\
		& & &  Value\\
		\hline
		$\mathcal{I}$.insert($I$) & Inserts interval $I$ into $\mathcal{I}$ & $O(\log(|\mathcal{I}|))$ worst case & -\\
		\hline
		$\mathcal{I}$.delete($I$) & Deletes interval $I$ from $\mathcal{I}$ & $O(\log(|\mathcal{I}|))$ worst case & -\\
		\hline
		$\mathcal{I}$.intersection($I$) & Returns a set of intervals $S_I$ in $\mathcal{I}$  & $O(\log(|\mathcal{I}|)+|S_I|)$ & $S_I$ \\
		&  that intersect with $I$ & worst case & \\
		\hline
		\hline
		\multicolumn{4}{|l|}{Doubly linked list $Q$ \cite{Cormen:linked-list}}\\
		\hline
		${Q}$.insert($x$) & Inserts element $x$ into list $Q$ & $O(1)$ worst case & -\\
		\hline
		${Q}$.delete($x$) & Deletes element $x$ from list $Q$ & $O(1)$ worst case & -\\
		\hline
		$Q$.begin() & Returns the first element $x$ of list $Q$ & $O(1)$ worst case & $x$ \\
		\hline
		\hline
		\multicolumn{4}{|l|}{Set ${U}$ \cite{book:cpp}}\\
		\hline
		${U}$.insert($x$) & Inserts a new element $x$ into $U$ & $O(1)$ amortized & -\\
		\hline
		$U$.begin() & Iterator to the first element of the set & $O(1)$ worst case & -\\
		\hline
		$U$.end() & Iterator to the last element of the set  & $O(1)$ worst case & - \\
		\hline
		\hline
		\multicolumn{4}{|l|}{Dynamic Array $A$ \cite{Cormen:2009:IAT:1614191}}\\
		\hline
		$A$.at($i$) & Inserts at $i$-th position of array $A$. & $O(1)$ amortized & -\\
		& Doubles the array size after initialization & & \\
		& if array is full && \\
		\hline
		\hline
		$A$.size() & Returns the size of array $A$ & $O(1)$ worst case & size\\
		\hline
		\hline
		\multicolumn{4}{|l|}{Red-Black Tree $R$ \cite{Cormen:RB-tree}}\\
		\hline
		$R$.insert($x$) & Inserts element $x$ into $R$ & $O(\log(|R|))$ worst case & -\\
		\hline
		$R$.delete($x$) & Deletes element $x$ from $R$& $O(\log(|R|))$ worst case & -\\
		\hline
		$R$.max() & Returns the maximum element $x$ in $R$ & $O(\log(|R|))$ worst case & $x$\\
		\hline
		$R$.min() & Returns the minimum element $x$ in $R$ & $O(\log(|R|))$ worst case & $x$\\
		\hline
		$R$.empty() & Checks if the tree $R$ is empty & $O(1)$ worst case & 0/1\\
		\hline
	\end{tabular}
	\caption{Data Structures used in dynamic setting}
	\label{table:1}
\end{table}

%%%%%%%%%%%%%%%%%%%%%%%%%%%%%%%%%%%%%%%%%%%%%%%%%%%%%
{\bf Interval Trees to store intervals and endpoints:}
\begin{enumerate}
	\item \textbf{Set of intervals $\mathcal{I}$.}
The set of intervals $\mathcal{I}$ is  maintained as an interval tree.  Therefore,  $\mathcal{I}$ is an interval tree such that for each $i \geq 1$, interval $I_i$ is maintained as its left and right endpoints $l_i$ and $r_i$.  Further, the level value and offset $L(I_i)$ and $o(I_i)$ are computed at the time of insertion, and updated as necessary in the fully dynamic case. The index of the update when $I_i$ is inserted is also stored and referred as time of insertion whenever necessary. 
 %{\bf Is the index of the update when $I_i$ is inserted also stored here?}
	% $\{l_i,r_i\}$, we also maintain the level  $L(I_i)$, offset $o(I_i)$, and time of insertion.   
	%%%%%%%%%%%%%%%%%%%%%%%%%%%%%%%%%%%%
%%%%%%%%%%%%%%%%%%%%%%%%%%%%%%%%%%%%%%%
	\item \textbf{Set of endpoints $\mathcal{E}$.} The set of endpoints of the intervals in $\mathcal{I}$ is stored as an interval tree denoted by $\mathcal{E}$. For every interval $I_i$=$[l_i,r_i]$, we maintain the left endpoint and the right endpoint as intervals $[l_i,l_i]$ and $[r_i,r_i]$ respectively in $\mathcal{E}$.
	\item \textbf{Hash table $T$ points to set of intervals with same level value.} For a non-negative integer $h$, $T[h]$ points to the interval tree which maintains the set of intervals with level value $h$.\\\\ 
	\end{enumerate}
	{\bf Data Structures to store supporting line segments:} At every endpoint $t \in \mathcal{E}$,  the SLS $e_t$ and the height, $h_t$, of $e_t$ is maintained.   In the incremental setting, the height of an SLS is non-decreasing with the updates and this need not be true in the fully dynamic case.  Thus we have different data structures to represent SLS in the incremental setting and the fully dynamic setting.  
%%%%%%%%%%%%%%%%%%%%%%%%%%%%%%%%%%%%%%%%%%%%%%%%%
\begin{enumerate}	
\item \textbf{Incremental setting: } In the incremental setting, the SLS $e_t$ is maintained using a  dynamic array $A_t$.  For a level value  $l$,  $A_t[l]$ is defined to be $1$ if there is an interval containing $t$ whose level value is $t$.  Otherwise, $A_t[l]$ is defined to be $0$.
Clearly, the height $h_t$ of the supporting line segment $e_t$ is the smallest index $l$ such that $A_t[l]=0$. 
To respond to queries for $h_t$ efficiently,  a doubly linked list $Q_t$ and a dynamic array $A^{\prime}_t$ are used as follows.
The head of the doubly linked list $Q_t$ contains the value of $h_t$, and following it, the set $\{l \mid A_t[l]=0\}$ as a doubly linked list in increasing order of the value of $l$.  
To maintain $h_t$, we define a doubly linked list $Q_t$ which stores every index $i$ in $A_t$ where $A_t[i]$ is $0$ in the increasing order of the value of $i$. Note that the value stored at the head node of $Q_t$ is $h_t$.  The dynamic array $A^{\prime}_t$ is defined as follows: for each $l \geq 0$,  if $A_t[l] = 1$, then $A^{\prime}_t[l]$ stores a pointer to the node in $Q_t$ which stores the index $l$; otherwise, $A^{\prime}_t[l] = NULL$.
% stores NULL$AWe augment $Q_t$ with another dynamic array $A^{\prime}_t$. For an index $i$, if $A_t[i]$ is $0$ then $A^{\prime}_t[i]$ stores a pointer to the node in $Q_t$ which stores the index $i$. If $A_t[i]$ is $1$ then $A^{\prime}_t[i]$ stores NULL. 
Using the dynamic array $A^{\prime}_t$, insert, delete, and search operations in $Q_t$ can be performed in constant time. 
%Since our algorithm is incremental, the dynamic arrays only expand. 
Insertion into a dynamic array takes amortized constant time \cite{Cormen:2009:IAT:1614191}. 
%The size of the array $A_t$ is the length of $A_t$. 
%During insertion, whenever size of $A_t$ is increased, size of  $A^{\prime}_t$ is also increased and appropriate nodes are inserted in $Q_t$. 
A query for the value of $h_t$ can be answered in constant time by returning the value stored in the head node of $Q_t$. 
%If $Q_t$  changes due to an insert, then we need to change the head node of $Q_t$ to the next node in the list and delete the previous head node of $Q_t$. This operation also takes constant time.
\item \textbf{Fully dynamic setting: } SLS $e_t$ is maintained using two  Red-Black trees, $Z_t$ and ${NZ}_t$. 
 A level value $l$ is stored in ${NZ}_t$ if there is an interval in $\mathcal{I}$ which contains $t$ and whose level value is $l$.
Otherwise, a level value $l$ is stored in $Z_t$.  
To compute the height $h_t$, we do the following: if $Z_t$ is non-empty, then the minimum value in $Z_t$ is the required height $h_t$. If $Z_t$ is empty, then $h_t$ is one more than the maximum value in ${NZ}_t$.  The number of nodes in the trees $Z_t$ and ${NZ}_t$ is at most  $\omega$. Therefore, the time required to compute $h_t$ is $O(\log \omega)$.
\end{enumerate}	
\noindent
\vspace{-0.9cm}
%%%%%%%%%%%%%%%%%%%%%%%%%%%%%%%%%%%%%%%%%%

%% file: Journal/Incremental.tex
% !TEX root = ../IntervalColoringFullVersion.tex
\section{Incremental Interval Coloring using Algorithm $\KTSLS$}
\label{sec: Incremental-Algorithm}
\noindent
%In the incremental setting, for each $1 \leq i \leq n$, $I_i = [l_i,r_i]$ represents the interval inserted in the $i$-th update step. Here $n$ is the total number of insertions. In this section, we design an incremental algorithm ($\Insertion$) for proper interval coloring which is an implementation of $\KTSLS$. 
%%%%%%%%%%%%%%%%%%%%%%%%%%%%%%%%%%%%%%%%%%%%%%%%%%%%%%%%%
We present the incremental algorithm $\Insertion$ which is an implementation of Algorithm $\KTSLS$. 
%\subsection{Algorithm $\Insertion$}
%\label{subsec:incrementalAlgo}
The pseudo code of $\Insertion$ is presented in Algorithm~\ref{alg:insert} along with the corresponding steps.   The procedures used in $\Insertion$ 
%are in Table~\ref{table:2} and they are 
described in Section~\ref{subssec:procedures-handle-insert}.   The amortized update time of $\Insertion$ is given by the Theorem~\ref{thm:IncrementalAlgo}.\\\\ 
\noindent
%Algorithm $\Insertion$ computes the color $(L(I_i),o(I_i))$ by implementing the steps as described below:\\
%%%%%%%%%%%%%%%%%%%%%%%%%%%%%%%%%%%%%%%%%
\begin{minipage}{.45\textwidth}
\begin{enumerate}
	\item {\bf Computing $L(I_i)$:}\\
	\textbf{Step 1}: Insert $I_i$ into the set of intervals $\mathcal{I}$ (Line 1 in Algorithm~\ref{alg:insert}). Check if the endpoint $l_i$ is already present in the set of endpoints $\mathcal{E}$ (Line 4 in Algorithm~\ref{alg:insert}). If not, then compute the SLS at endpoint $l_i$ (Line 5 in Algorithm~\ref{alg:insert}) and insert $l_i$ into the set $\mathcal{E}$ (Line 6 in Algorithm~\ref{alg:insert}). Repeat  for endpoint $r_i$ (Line 8-11 in Algorithm~\ref{alg:insert}).\\
	\textbf{Step 2}: Compute the set $S = \mathcal{E} \cap I_i$.  For each $t \in S$, compute $h_t$, the height of the SLS $e_t$ at $t$.  
	Assign $L(v_i) = \displaystyle \max_{t \in S} h_t$ (Line 12-13 in Algorithm~\ref{alg:insert}).\\
	\textbf{Step 3}: Update the SLS $e_t$ for each point $t \in S$ (Line 14 in Algorithm~\ref{alg:insert}). 
	%%%%%%%%%%%%%%%%%%%%%%%%%%%%%%%%%%%%%%%%%
	\item {\bf Computing $o(I_i)$}:
	Compute $o(I_i)$ to be the smallest value from the set $\{1,2,3\}$ which is different from the offset of the neighbours of $I_i$ which have the level $L(I_i)$ (Line 15 in Algorithm~\ref{alg:insert}). 
\end{enumerate}
\end{minipage}
\hspace{.6cm}
%%%%%%%%%%%%%%%%%%%%%%%%%%%%%%%%%%%%%%%%%%%
% We analyze update time of $\Insertion$ in Section~\ref{subsection : Incremental-Algorithm-Analysis}.  
%%%%%%%%%%%%%%%%%%%%%%%%%%%%%%%%%%%%%%%%%%%%%
\begin{minipage}{.45\textwidth}
\begin{algorithm}[H]	
	\caption{$\Insertion$($I_i = [l_i,r_i]$) is used to handle insertion of interval $I_i$}
	\begin{algorithmic}[1]
		\State $\mathcal{I}$.insert($I_i$) 
		\State $I_i^{l} \gets$ $[l_i,l_i]$   
		\State $I_i^{r} \gets$ $[r_i,r_i]$   
		\If{($\mid \mathcal{E}$.intersection($I_i^{l}) \mid$ = $0$)} 
		\State $A_{l_i}, Q_{l_i} \leftarrow$  \Call{$\getSLS$}{$\mathcal{I}$,$l_i$} 
		\State $\mathcal{E}$.insert($I_i^{l}$) 
		\EndIf
		\If{($\mid \mathcal{E}$.intersection($I_i^{r}) \mid$ = $0$)}  
		\State $A_{r_i}, Q_{r_i} \leftarrow$  \Call{$\getSLS$}{$\mathcal{I}$,$r_i$} 
		\State $\mathcal{E}$.insert($I_i^{r}$) 
		\EndIf
		\State $S, h \leftarrow$  \Call{$\computemaxheightSLS$}{$\mathcal{E}$,$I_i$}
		\State $L(I_i) \leftarrow h$
		\State \Call{$\updateEndPoints$}{$S,L(I_i)$}
		\State \Call{$\OFFSET$}{$I_i$}
	\end{algorithmic}
	\label{alg:insert}
\end{algorithm}
\end{minipage}
%%%%%%%%%%%%%%%%%%%%%%%%%%%%%%%%%%%%%%%%%%%
\noindent
%%%%%%%%%%%%%%%%%%%%%%%%%%%%%%%%%%%%%%%%%%%%
%%%%%%%%%%%%%%%%%%%%%%%%%%%%%%%%%%%%%%%%%%%%%%
\begin{theorem}
	\label{thm:IncrementalAlgo}
	$\Insertion$ is an incremental algorithm which supports insertion of 	
	a sequence of $n$ intervals in amortized $O(\log n + \Delta)$ time per update.
\end{theorem}
\begin{proof}
We analyze the running time for computing $L(I_i)$ and $o(I_i)$.  \\
	%We analyze $\Insertion$ by analyzing the time required in every step of the algorithm. 
 \textbf{Analysis for computing $L(I_i)$:} Computing $L(I_i)$ involves $3$ steps.
		\begin{enumerate}
			\item Step 1 in computing $L(I_i)$ takes $O(\log n + \omega)$ time: Insertion of $I_i$ = $[l_i,r_i]$ into $\mathcal{I}$ takes $O(\log n)$ time. Let $I_i^{l}$ = $[l_i,l_i]$. Checking if $I_i^{l}$ is present in $\mathcal{E}$ by an intersection query takes $O(\log n)$ time in the worst case.     If $I_i^{l}$ is in $\mathcal{E}$, then no further processing is done. On the other hand,  if $I_i^{l}$ is not in $\mathcal{E}$ then procedure $\getSLS$($\mathcal{I},l_i$) is invoked. From Lemma~\ref{lem:runningtime-computeSLS}, procedure $\getSLS$ takes $O(\log n + \omega)$ time.  The same steps are repeated for $I_i^{r}$ = $[r_i,r_i]$. Hence Step 1 in computing $L(I_i)$ for interval $I_i$ takes $O(\log n + \omega)$ time.
					%%%%%%%%%%%%%%%%%%%%%%%%%%%%%%%%%%%%%%%%%%%%%
			\item Step 2 in computing $L(I_i)$ takes $O(\log n + \Delta)$ time: Procedure $\computemaxheightSLS$$(\mathcal{E},I_i)$ is invoked to perform this step. From Lemma~\ref{lem:runningTimeMaxHeightofSLS}, procedure $\computemaxheightSLS$(Algorithm~\ref{alg:MAX-HEIGHT-OF-SLS-IN-INTRVAL-usingDynamicArray}) takes $O(\log n + \Delta)$ time. Hence Step 2 in computing $L(I_i)$ for interval $I_i$ takes $O(\log n + \Delta)$ time. 
			%%%%%%%%%%%%%%%%%%%%%%%%%%%%%%%%%%%%%%%%%%%%%%
			\item Step $3$ in computing $L(I_i)$ takes amortized $O(\Delta)$ time: Procedure $\updateEndPoints$$(S,L(I_i))$ is invoked to perform this step. From Lemma~\ref{lem:runningTimeUpdateEndpoints}, procedure $\updateEndPoints$(Algorithm~\ref{algo:UPDATE-EDGE-POINTS-usingDynamicArray}) takes amortized $O(\Delta)$ time. Hence Step $3$ in computing $L(I_i)$ for interval $I_i$ takes amortized $O(\Delta)$ time. 
		\end{enumerate}
\textbf{Analysis for computing $o(I_i)$:}
		To compute the offset value of interval $I_i$ with level value $L(I_i)$, procedure $\OFFSET$$(I_i)$ is invoked. From Lemma~\ref{lem:time-greeyd-coloring}, procedure $\OFFSET$$(I_i)$ takes $O(\log n)$ time.\\
	%%%%%%%%%%%%%%%%%%%%%%%%%%%%%%%%%%%%%%%%%%%%%%%%%%%%	
	Therefore, total time taken by $\Insertion$ for insertion of $n$ intervals is the total time taken for Step 1, Step 2, Step 3, and the total time spent in computing the offset.   For interval graphs, it is well known that $\omega=\chi \leq\Delta+1$.  Thus the  running time is $O(n\log n + n \Delta)$.  Therefore, the amortized update time over a sequence of $n$ interval insertions is  $O(\log n + \Delta)$. Hence the Theorem.
%	$\mathcal{T}$ = Total time for Step1 + Total Time for Step 2 + Total Time for Step 3 + Total time for computing offset \\
%	$\mathcal{T}$ = $n \times O(\log n + \omega)$ + $n \times O(\log n + \Delta)$ + $n \times O(\Delta)$ + $n \times O(\log n)$\\
%	$\mathcal{T}$ = $O(n\log n + n \Delta)$\\
%	Therefore, the amortized update time over a sequence of $n$ interval insertions is  $O(\log n + \Delta)$. 
\end{proof}
%%%%%%%%%%%%%%%%%%%%%%%%%%%%%%%%%%%%%%%%%%%%%%
\subsection{Procedures used in $\Insertion$}
\label{subssec:procedures-handle-insert}
\noindent
\begin{minipage}{.45\textwidth}
\vspace{2mm}	
The data structures used in designing these procedures are listed in Table~\ref{table:1}.	
	\begin{lemma}
		\label{lem:time-greeyd-coloring}	
		Procedure $\OFFSET$ takes as input interval $I$, computes the offset value for interval $I$ and takes $O(\log n)$ time in the worst case.
	\end{lemma}
	\begin{proof}
		Hash table $T$ is used to access the interval tree $T[L(I)]$. If $T[L(I)]$ is NULL, then an interval tree $T[L(I)]$ is created with $I$ as the first interval (Line 2-3 in Algorithm~\ref{alg:offset}). Otherwise, $T[L(I)]$ is the interval tree which stores all the intervals with level value same as $L(I)$. An intersection query is performed on $T[L(I)]$ with $I$ to obtain all the intervals that intersect with $I$ (Line 5 in Algorithm~\ref{alg:offset}). From Property {\bf P}, the maximum number of intervals returned by the above query is $2$. The offset value of interval $I$, $o(I)$, is set to be the smallest value from $\{1,2,3\}$ not assigned to any of the at most two neighbors of $I$ in level $L(I)$ (Line 6 in Algorithm~\ref{alg:offset}). Interval $I$ is inserted to $T[L(I)]$ (Line 7 in Algorithm~\ref{alg:offset}).\par 
		Running time of $\OFFSET$ is dominated by intersection query in Line 5 and insertion of interval $I$  in Line 7. Since $|\mathcal{I}| \leq n$, worst case running time of $\OFFSET$ is $O(\log n)$. Hence the Lemma.
	\end{proof}
\end{minipage}
\hspace{0.6cm}
\begin{minipage}{.45\textwidth}
	\begin{algorithm}[H]
		\caption{$\OFFSET$($I$) is used to compute the offset value of the interval $I$ with level value $L(I)$. }
		\begin{algorithmic}[1]
			\Procedure{$\OFFSET$}{$I$}
			\If{$T[L(I)]$ \textbf{is} \textit{NULL}}
			\State $T[L(I)] \leftarrow $ \Call{Interval\_Tree}{\ }
			\EndIf
			\State S $ \leftarrow T[L(I)].$intersection($I$)
			\State $o(I) \leftarrow$  The minimum value in the set $\{1, 2, 3\}$ which is not the offset value of any interval in S. 
			\State $T[L(I)].$insert($I$)
			\EndProcedure
		\end{algorithmic}
		\label{alg:offset}
	\end{algorithm}
\end{minipage}
%%%%%%%%%%%%%%%%%%%%%%%%%%%%%%%%%%%%%%%%%

\noindent
\begin{minipage}{.45\textwidth}	
\begin{comment}  	
\textbf{Procedure $\getSLS$($\mathcal{I}_t, t$):} This procedure is used to compute the SLS at endpoint $t$ and is described in Algorithm~\ref{alg:GET-SLS-dynamicArray}.  It performs an intersection query on $\mathcal{I}$ with interval $[t,t]$. Let $\mathcal{I}_t$ denote the set returned by the intersection query. Dynamic arrays $A_{t}$ and $A^{\prime}_{t}$, each of size $\max(levels(\mathcal{I}_{t}))$, are created. For every $i$ in the range $[0, \max(levels(\mathcal{I}_{t}))]$, $A_t[i]$ is set to $1$ if $i \in levels(\mathcal{I}_t)$ and $A_t[i]$ is set to $0$ otherwise. For every $A_t[i]$ = $0$, a node storing the index $i$ is inserted to the doubly linked list $Q_t$ and pointer to that node is stored in $A^{\prime}_t[i]$. For every $A_t[i]$ = $1$, we store a NULL in $A^{\prime}_t[i]$.
%\centering
\end{comment}
\begin{lemma}
\label{lem:runningtime-computeSLS}
Procedure $\getSLS$ takes set of intervals $\mathcal{I}$ and endpoint $t$ as input, maintains SLS at $t$ using dynamic array $A_t$ and doubly linked list $Q_t$, and takes $O(\log n + \omega)$ time in the worst case. 
\end{lemma}
\begin{proof}
$\getSLS$ performs an intersection query on $\mathcal{I}$ with interval $[t,t]$ (Line 2 in Algorithm~\ref{alg:GET-SLS-dynamicArray}).  	
Let $\mathcal{I}_t$ denote the set returned by the intersection query. Maximum height $h_t$ = $\max(levels(\mathcal{I}_{t}))$ and the set $U$ = $levels(\mathcal{I}_t)$  are computed (Line 7-9 in Algorithm~\ref{alg:GET-SLS-dynamicArray}). 
For every $i$ in the range $[0, h_t]$, $A_t[i]$ is set to $0$ and $i$ is inserted to $Q_t$ (Line 11-13 in Algorithm~\ref{alg:GET-SLS-dynamicArray}).  
For every $i$ in the range $[0, h_t]$, $A_t[i]$ is reset to $1$ if $i \in levels(\mathcal{I}_t)$ and
$i$ is deleted from $Q_t$ (Line 15-17 in Algorithm~\ref{alg:GET-SLS-dynamicArray}). It returns $A_t$ and $Q_t$.\par
Running time of $\getSLS$ is dominated by the intersection query in Line 2, loop in Line 7-9, loop in Line 11-13, and loop in Line 15-17. At any level, SLS ${e_t}$ intersects with at most $2$ intervals and we have $\omega$ many levels. Hence, $|\mathcal{I}_{t}|$ = $O(\omega)$, $h_t \leq \omega + 1$, and $|U| \leq \omega$. Again, $|\mathcal{I}| \leq n$. Therefore, time taken by procedure $\getSLS$ in the worst case is $O(\log n + \omega)$. Hence the Lemma.  
\end{proof}
%%%%%%%%%%%%%%%%%%%%%%%%%%%%%%%%%%%%%%%%
\begin{lemma}
\label{lem:runningTimeMaxHeightofSLS}
Procedure $\computemaxheightSLS$ takes set of endpoints $\mathcal{E}$ and interval $I$ as input, computes the set of endpoints $S$ contained in $I$ and maximum among the height of the SLS at the endpoints in $S$, and takes $O(\log n + \Delta)$ time in the worst case.
\end{lemma}
\begin{proof}
Procedure $\computemaxheightSLS$ performs an intersection query of $I$ on $\mathcal{E}$ (Line 2 in Algorithm~\ref{alg:MAX-HEIGHT-OF-SLS-IN-INTRVAL-usingDynamicArray}). This query returns the set $S$ of all the endpoints which intersect with $I$. It computes the maximum among the  height of the SLS at  endpoints in $S$, denoted by $h$ (Line 4-11). The procedure returns $h$ and the set $S$.\par 	
Running time of $\computemaxheightSLS$ is dominated by intersection query in Line 2 and loop from Line 4-11. Since $\Delta$ is the maximum degree in the associated interval graph, interval $I$ can intersect with at most $\Delta$ intervals. Therefore, $|S|$ = $O(\Delta)$. Further, $|\mathcal{E}| \leq 2n$. Thus the worst case time taken by $\computemaxheightSLS$ is $O(\log n + \Delta)$. Hence the Lemma.	   	
\end{proof}
\end{minipage}
   \hspace{.5cm}
   \begin{minipage}{.45\textwidth}
  The different procedures used in $\Insertion$ are described and analyzed below. \\
   \begin{algorithm}[H]
  \caption{ $\getSLS$($\mathcal{I}$,$t$) computes the supporting line segment at endpoint $t$. }
  %Here the SLS is maintained using a dynamic array and doubly linked list.}
	\begin{algorithmic}[1]
		\Procedure{$\getSLS$}{$\mathcal{I}$,$t$}
		\State $\mathcal{I}_t \leftarrow \mathcal{I}$.intersection($[t,t]$)
		\State $h_t \leftarrow 0$
		\State $U \leftarrow$ Empty Set	
		\State $A_t \leftarrow [\ ]$
		\State $Q_t \leftarrow $ Empty doubly linked list
		\For{$I$ \textbf{in} $\mathcal{I}_t$}
		\State $h_t \leftarrow \max(h_t,L(I))$
		\State $U$.insert$(L(I))$
		\EndFor
		\For{$i$ \textbf{in} $\{0,1,2,...h_t\}$} 
		\State $A_t$.at$(i) \leftarrow $  0
		\State $Q_t$.insert$(i)$
		\EndFor
		\For{$i$ \textbf{in} $U$.begin() \textbf{to} $U$.end()} 
		\State $A_t.$at$(i) \leftarrow$  1
		\State $Q_t$.delete$(i)$
		\EndFor
		\State \Return $A_t, Q_t$
		\EndProcedure
	\end{algorithmic}
	\label{alg:GET-SLS-dynamicArray}
	\end{algorithm}
%%%%%%%%%%%%%%%%%%%%%%%%%%%%%%%%%%%%%%%%%%	
 \begin{algorithm}[H]
     \caption{$\computemaxheightSLS$($\mathcal{E}$,$I$) computes the set $S$ of endpoints contained in interval $I$ and the maximum value among the heights of the SLS at these points.}
     % compute the maximum height $h$ of all the supporting line segments at these endpoints. In this procedure SLS is maintained using dynamic array and doubly linked list.}
	\begin{algorithmic}[1]
		\Procedure{$\computemaxheightSLS$}{$\mathcal{E}$,$I$} 
		\State $S \leftarrow  \mathcal{E}$.intersection($I$)
		\State $h \leftarrow 0$
		\For{$t$ \textbf{in} $S$}
		\If{$Q_t$.begin() is NULL}
		\State $h_t \leftarrow $0
		\Else 
		\State $h_t \leftarrow {Q_t}$.begin() 
		\EndIf
		\State $h \leftarrow \max(h,h_t)$
		\EndFor
		\State return $S, h$
		\EndProcedure
	\end{algorithmic}
	\label{alg:MAX-HEIGHT-OF-SLS-IN-INTRVAL-usingDynamicArray}
	\end{algorithm}
	\end{minipage}
	%\vspace{-3cm}
	        
   \noindent
 \begin{minipage}{.45\textwidth}
  \textbf{Procedure $\updateEndPoints$$(S,L(I))$: } This procedure, described in Algorithm~\ref{algo:UPDATE-EDGE-POINTS-usingDynamicArray}, is used to update the SLS at the endpoints contained in set $S$: let $l = L(I)$. For every endpoint $t \in S$, size of $A_t$ is checked.  
\begin{enumerate}
	\item 
	%\label{case:1}
	{\bf Case A:} If $l < {A_t}$.size(). In this case, $A_t[l]$ is set to $1$. The pointer stored in $A^{\prime}_t[l]$ is used to delete the node in $Q_t$ which stores the value $l$ and $A^{\prime}_t[l]$ is set to NULL subsequently. If the deleted node in $Q_t$ was the head node, then the head node is updated to the next node in $Q_t$ and thus the value of $h_t$ also gets updated.  
	%%%%%%%%%%%%%%%%%%%%%%%%%%%%%%%%%%%%
	\item 
	%\label{case :2}
	{\bf Case B:} If $l \geq {A_t}$.size(). In this case, the standard doubling technique for expansion of dynamic arrays \cite{Cormen:2009:IAT:1614191} is used to increase the size of $A_t$ until ${A_t}$.size() becomes strictly greater than $l$.  $A^{\prime}_t$ is also expanded along with $A_t$ and appropriate nodes are inserted to $Q_t$. Once ${A_t}$.size() $> l$, the remaining operations are same as in the case \textbf{A}.  
	\end{enumerate}
  \begin{algorithm}[H]
	\caption{$\updateEndPoints$($S$,$L(I)$) is used to update the supporting line segments at the endpoints contained in set $S$ for level value $L(I)$. In this procedure, SLS is maintained using dynamic array and doubly linked list.}
	\begin{algorithmic}[1]
		\Procedure{$\updateEndPoints$}{$S$,$L(I)$}
		\State $l \leftarrow L(I)$ 
		\For{$t$ \textbf{in} $S$}
		\If{$l > A_t$.size()}
		\For{$i$ \textbf{in} $\{A_t$.size()$+1, l\}$ }
		\State {$A_t$.at$(i) \leftarrow$ 0}
		\State $Q_t$.insert$(i)$
		\EndFor
		\EndIf
		\State{$A_t$.at$(l) \leftarrow$  1}
		\State{$Q_t$.delete$(l)$}
		\State{$h_t \leftarrow Q_t$.begin()} 
		\EndFor
		\EndProcedure
	\end{algorithmic}
	\label{algo:UPDATE-EDGE-POINTS-usingDynamicArray}
\end{algorithm}
       \end{minipage}
  \hspace{.6cm}
  \begin{minipage}{.45\textwidth}
 
	\begin{lemma}
\label{lem:runningTimeUpdateEndpoints}
Procedure $\updateEndPoints$ takes $O(\Delta)$ amortized time.
\end{lemma}
\begin{proof}
For every endpoint $t \in S$, size of $A_t$ is checked in constant time. To analyze the time required in $\updateEndPoints$, we observe that every update must perform the operations as described in case \textbf{A}. We refer to these operations as \textit{task M}(M stands for mandatory). Some updates have to perform additional operations as described in case \textbf{B}. We refer to these operations as \textit{task A}(A stands for additional). 
The time taken by each update to perform \textit{task M} is $|S| \times O(1) = O(|S|)$. Since $\Delta$ is the maximum degree, it follows that $|S| \leq \Delta$. Therefore, every update takes $O(\Delta)$ time to perform \textit{task M} in the worst case. To analyze the time required to perform \textit{task A}, we crucially use the fact that our algorithm is incremental and hence only expansions of the dynamic arrays take place. Since $\omega$ is the size of the maximum clique, it follows that the maximum size of a dynamic array throughout the entire execution of the algorithm is upper bounded by $2\omega$. Over a sequence of $n$ insertions, the total number of endpoints is upper bounded by $2n$. Therefore, we maintain at most $4n$ dynamic arrays. For every such array,  total number of inserts in the array and the associated doubly linked list is at most $2 \omega$ in the entire run of the algorithm. An insertion into the dynamic array takes constant amortized time and insertion into doubly linked list takes constant worst case time. Therefore, during the entire run of the algorithm total time required to perform \textit{task A} on one dynamic array and its associated doubly linked list is $O(\omega)$. This implies that during the entire run of the algorithm total time spent on \textit{task A} over all the updates is $\leq 4n \times O(\omega)$. Let ${\sf T}$ be the total time spent on $\updateEndPoints$ at the end of $n$ insertions. This is the sum of the total time for \textit{task A} and the total time for \textit{task M}. Further, since $\omega \leq \Delta + 1$, it follows that  ${\sf T} = O(n \Delta)$. Hence the Lemma. 
%Therefore, \\
%${\sf T} \leq 4n \times O(\omega) + n \times O(\Delta)$\\
%${\sf T} \leq 4n \times O(\Delta) + n \times O(\Delta)$ [since $\omega \leq \Delta + 1$]\\
%Therefore, ${\sf T} = O(n \Delta)$. Hence the Lemma.
\end{proof}
\end{minipage}
\vspace{2cm}

%% file: Journal/FullyDynamic.tex
% !TEX root = ../IntervalColoringFullVersion.tex
\section{Fully dynamic interval coloring}
\label{sec:fully-dynamic-algorithm}
\noindent
An update in the update sequence in the fully dynamic setting consists of an interval to be colored  or a previously colored interval to be deleted.  The $i$-th update is {\sf Insert$(I_i)$} where $I_i$ is the interval presented to the algorithm.  The update {\sf Delete$(I_i)$} is to delete the  interval $I_i$ that was inserted during the $i$-th update.   At the end of each update, the invariants are maintained such that it follows that the intervals are colored with at most $3 \omega - 2$ colors, where $\omega$ is the size of the maximum clique in the interval graph just after the update.  
%At the end of each update, each interval $I$ is assigned a color $(L(I),o(I))$, such that $L(I) \leq \omega-1$ and property {\bf P} is satisfied.    {\bf After every update the algorithm ensures that for each interval $I$, $L(I)$ is the smallest value such that there exists an endpoint $t$ contained in $I$ such that  the height of the SLS at $t$ is at least $L(I)$.}
For an insert update, 
%
%update for insertion of a new interval and deletion of a previously colored interval. In the fully dynamic setting, $I_i$ represents the interval inserted in the $i$-th update step and value of $i$ ranges from $1$ to {\sf U}, where {\sf U} is the total number of updates. In the update sequence, an update of type {\sf Insert$(I_i)$} signifies that the current update is the $i$-th update which is presenting a new interval to the algorithm for insertion. An update of type {\sf Delete$(I_i)$} signifies that the interval $I_i$ with color $(L(I_i), o(I_i))$, which was inserted previously in the $i$-th update, is presented to the algorithm for deletion.\par
we use $\Insertion$ (Algorithm~\ref{alg:insert}) to ensure that the invariants are maintained at the end of the update.  However, to get a good bound on the update time, we use a different set of data structures to maintain SLS (see Section~\ref{subsec:DSpreliminaries}).    Therefore, the major result in this section is to handle the delete of a previously colored interval.  There are two aspects in the algorithm:  the first one is to ensure that the invariants are maintained after a delete, and the second one is to ensure that the update is efficient.  
%To handle the deletion of an interval $I_i$ (recall, that this is the interval inserted in the $i$-th update), we use Algorithm $\Deletion$.
%%%%%%%%%%%%%%%%%%%%%%%%%%%%%%%%%%%%%%%%%%%%%%%%%
\subsection{Algorithm $\Deletion$ for {\sf Delete$(I_i)$}}
\label{subsec:deletionAlgo} 
\noindent
%Let $I_i = [l_i,r_i]$ with color $(L(I_i),o(I_i))$ be the interval presented for deletion in the current update step.
Let $(L(I_i),o(I_i))$ be the color of $I_i$ at the beginning of the update.   The pseudo code and steps of Algorithm $\Deletion$ are presented in Algorithm~\ref{alg:delete}. \\
%and a summarized description of the procedures used in Algorithm~\ref{alg:delete} in Table~\ref{table:3}. 
%Algorithm $\Deletion$ supports deletion of an interval by implementing the steps as described below:\\
%%%%%%%%%%%%%%%%%%%%%%%%%%%%%%%%%%%%%%%%%%%%%%%%%
\begin{minipage}{.45\textwidth}
\vspace{2mm}	
{\bf Step 1:} Remove interval $I_i$ from the set of intervals $\mathcal{I}$ and hash table $T[L(I_i)]$.(Line 1 and Line 2 in Algorithm~\ref{alg:delete})\\
{\bf Step 2:} Compute the set of endpoints contained in $I_i$. Let $\mathcal{E}_i$ = $\mathcal{E} \cap I_i$. (Line 3 in Algorithm~\ref{alg:delete})\\
\textbf{Step 3:} For each endpoint $t \in \mathcal{E}_i$, update SLS $e_t$  to reflect the deletion of interval $I_i$. (Line 4 to 8 in Algorithm~\ref{alg:delete})\\
\textbf{Step 4:} Compute the set of intervals intersecting with $I_i$ and with level value strictly greater than $L(I_i)$. Let 
$\D$ = $\{I | L(I) > L(I_i), I_i \cap I \neq \phi\}$ (Line 10 in Algorithm~\ref{alg:delete}). Sort $\D$ in the increasing order of level value and break the ties in the increasing order of time of insertion.
For every interval $I$ in $\D$, repeat the following steps:
\begin{itemize}
	\vspace{-3mm}
	\item Compute the set of endpoints intersecting with $I$. Let $S$ = $\mathcal{E} \cap I$. For every endpoint $t \in S$ compute the height $h_t$ of the SLS $e_t$. Compute $h$ = $\max\{h_t | t \in S\}$ (Line 11 in Algorithm~\ref{alg:delete}).
	\vspace{-2.5mm}
	\item If $h \geq L(I)$ then no further processing is required for interval $I$. 
	%Then $L(I_j)$ is unchanged. Prior to this update, we maintained invariant $\C$ for interval $I_j$. Invariant $\C$ persists for $I_j$ after the update.
	\vspace{-2.5mm}
	\item If $h < L(I)$ then following steps are executed (Line 17 to 25 in Algorithm~\ref{alg:delete}):
	%, Invariant $\C$ is violated for $I_j$. Following steps are executed to restore Invariant $\C$ for $I_j$ (Line 17 to 25 in Algorithm~\ref{alg:delete}):
	\begin{itemize}
		\vspace{-2.5mm}
		\item Change level value of $I$ from $L(I)$ to $h$.   
		\item Recompute the offset value $o(I)$ for $I$ with the new level value $h$.
		\item Update SLS $e_t$ for every point $t \in S$ to reflect the change in level value of $I$. 
	\end{itemize}
\end{itemize}
\end{minipage}
\hspace{0.6cm}
%%%%%%%%%%%%%%%%%%%%%%%%%%%%%%%%%%%%%%
%%%%%%%%%%%%%%%%%%%%%%%%%%%%%%%%%%%%%%%%%%%%%%%%%%%%%%%%%%%%%%%
\begin{minipage}{.45\textwidth}
\begin{algorithm}[H]	
	\caption{$\Deletion$($I_i = [l_i,r_i]$) is used to handle deletion of interval $I_i$}
	\begin{algorithmic}[1]
		\State $T[L(I_i)].$delete$(I_{i})$ 
		\State $\mathcal{I}$.delete($I_{i}$)
		\State $\mathcal{E}_i \leftarrow \mathcal{E}$.intersection($I_{i}$)
		\For{$t$ \textbf{in} $\mathcal{E}_i$}
		\If{$\mid T[L(I_i)]$.intersection($t) \mid$ $ = 0$}
		\State ${NZ}_t$.delete($L(I_i)$ 
		\State  $Z_t$.insert($L(I_i)$)
		%\State $Z_t$.insert($L(I_i)$)
		\EndIf
		\EndFor
		\State{Compute $\D \leftarrow \{I | L(I) > L(I_i), I_i \cap I \neq \phi\}$. Sort $\D$ in the increasing order of level value and break the ties in the increasing order of time of insertion.} 
		\For{$I$ \textbf{in} $\D$}
		\State $S, h \leftarrow $ \Call{$\computemaxheightSLS$}{$\mathcal{E}$,$I$}
		\If{$h \geq L(I)$} 
		\State \textbf{continue}
		\EndIf
		\State $T[L(I)].$delete($I$)
		\State $T[h].$insert($I$)
		\For{$t$ \textbf{in} $S$}
		\If{$\mid T[L(I)]$.intersection($t) \mid$ $= 0$}
		\State ${NZ}_t$.delete($L(I)$ 
		\State  $Z_t$.insert($L(I)$)
		\EndIf
		\State {${NZ}_t$.insert($h$) 
		\State $Z_t$.delete($h$)}
		\EndFor
		\State $L(I) \leftarrow h$
		\State \Call{$\OFFSET$}{$I$}
		\EndFor
		
	\end{algorithmic}
	\label{alg:delete}
\end{algorithm}
\end{minipage}
%%%%%%%%%%%%%%%%%%%%%%%%%%%%%%%%%%%%%%%%%%%%%%%%%%%%%%%%%%%%%%%%%%%%%%
\subsection{Correctness of  $\Deletion$}
\label{subsec:correctnessDeletionAlgo}
\noindent
The first crucial property maintained by Algorithm $\Deletion$ is that at the end of the update, for each interval $I$, there exists a point $t$ in $I$ such that the height of the SLS at $t$ is at least $L(I)$.  We refer to this property as  {\bf Invariant $\C$}.  The second crucial property maintained is property {\bf P}.  The following lemma proves a bound on the number of colors used.
%These two properties are also maintained after Algorithm $\Insertion$ (Algorithm~\ref{alg:insert}).
%In the case of  deletion of an interval, it is crucial to ensure that level value assigned to any new interval after a deletion is at most $\omega - 1$ and such an interval satisfies property {\bf P}. In order to achieve the same, every interval $I_j$ which remains after a deletion must satisfy the following invariant:\\    
%%%%%%%%%%%%%%%%%%%%%%%%%%%%%%%%%%%%%%%%%%
%{\bf Invariant $\C$:} An interval $I_j$ is said to satisfy Invariant $\C$ if there exists an endpoint $t$ contained in $I_j$ such that the height of the SLS at $t$ is at least $L(I_j)$.\par 
%%%%%%%%%%%%%%%%%%%%%%%%%%%%%%%%%%%%%%%%%%
%An interval is called \textit{dirty} after a deletion if invariant $\C$ is violated for that interval. An interval which is not dirty after a deletion is called \textit{clean}.
%Let $\D$ denotes the set of all dirty intervals after the deletion of interval $I_i$.
%Lemma~\ref{lem:3-competitive-fully-dynamic} shows that maintaining invariant $\C$ ensures property {\bf P}. 
%%%%%%%%%%%%%%%%%%%%%%%%%%%%%%%%%%%%%%%%%%%%%
\begin{lemma}
	\label{lem:3-competitive-fully-dynamic}
	At the end of each update, the number of colors used is  at most $3 \omega - 2$, where $\omega$ is the size of the maximum clique in the interval graph just after the update.	
\end{lemma}
\begin{proof}
We start by assuming that prior to the update the invariant $\C$ and Property {\bf P} is satisfied by the coloring.  We show that after the update they continue to be satisfied.
For update {\sf Insert$(I_i)$} invariant $\C$ and Property {\bf P} are satisfied by the coloring at the end of the update.  This follows from Lemma \ref{lem1: max-height-equal-p(v)} which shows that invariant $\C$ is maintained by Algorithm $\Insertion$ and  Lemma~\ref{lem: Property-P} which proves that Algorithm $\Insertion$ maintains Property {\bf P}. 

For update {\sf Delete$(I_i)$} we know from Lemma~\ref{lem:set-of-dirty-intervals} that the set $\D$ consists of those intervals whose level values are changed by Algorithm $\Deletion$ to ensure that invariant $\C$ is maintained.  
Algorithm $\Deletion$ iterates over each interval $I$ in $\D$ and ensures, by modifying $L(I)$ if necessary, that there is a point $t \in I$ such that
$h_t \geq L(I)$.  Whenever $L(I)$ is modified, it is modified to be the maximum $h_t$, over all $t \in I$.   This choice of $L(I)$ also ensures that 
$I$  has at most two neighbors in the level $L(I)$ and none of the neighbors with level number $L(I)$ has a containment relationship with $I$.  The proof uses the same argument in Lemma~\ref{lem: Property-P}.  Thus Property {\bf P} is maintained by Algorithm $\Deletion$.

%by contradiction by assuming that $I_j$ is the first interval for which this claim is false.  The contradiction is obtained by using the argument in Lemma~\ref{lem: Property-P}.\par 
%Further,  for each $I_j \in \D$  the algorithm maintains the following invariant that  whenever $L(I_j)$ is modified: 
%Therefore, Algorithms $\Deletion$ maintains invariant $\C$.\par %that after every update, for each interval $I_i$, there is a point $t \in I_i$ such that $h_t \geq L(I_i)$.\par  
%For $\Deletion$, invariant $\C$ ensures the largest level value of any interval is at most $\omega - 1$.  Further, the intervals whose level values are $0$ form an independent set.     
	%Further,  for each $I_j \in \D$  the algorithm maintains the following invariant that  whenever $L(I_j)$ is modified: 
	%$L(I_j)$ is modified to be the maximum $h_t$, over all $t \in I_j$.  
Therefore, it follows that after the update, Property ${\bf P}$ is satisfied by the level values of the intervals and from invariant $\C$ it follows that the largest level value of any interval is at most $\omega - 1$.    Further, the intervals whose level values are $0$ form an independent set.     Therefore, the number of colors used by the algorithm at the end of an update is $3 \omega - 2$.  
 Consequently, the algorithm uses at most $3 \omega - 2 $ colors after each update step. Hence the Lemma.
\end{proof}
\noindent
%%%%%%%%%%%%%%%%%%%%%%%%%%%%%%%%%%%%%%%%%%%%%%%%%%%%%%%%%%%%%%%%%%%%%%%%%%%%%%%%%%%%%%
We now prove that on the update {\sf Delete$(I_i)$}, it is sufficient for Algorithm $\Deletion$ to consider the set $\D$ which is defined to be $\D$ = $\{I | L(I) > L(I_i), I_i \cap I \neq \phi\}$.  
On update {\sf Delete$(I_i)$}, an interval is called \textit{dirty} if during execution of Algorithm $\Deletion$  invariant $\C$ is violated for that interval. 
An interval which is not dirty during the execution of Algorithm $\Deletion$  is called \textit{clean}.
In Lemma~\ref{lem:set-of-dirty-intervals}, we show that $\D$ is a super set of all such intervals which become dirty. 
%
%The efficiency of the deletion update depends on the efficient computation of the set $\D$ and clean-up of $\D$ so that at the end of the update, there are no dirty intervals.  We also must ensure that during the clean-up of the set $\D$ no clean interval becomes dirty. On deletion of interval $I_i$, we initialize the set $\D$ = $\{I_j | L(I_j) > L(I_i), I_i \cap I_j \neq \phi\}$. In other words, the set $\D$ consists of every interval $I_j$ such that $L(I_j) > L(I_i)$ and $I_j$ intersects with $I_i$. Then, $\D$ is sorted in the increasing order of level value of the intervals and ties are broken in the increasing order of time of insertion.  
%Lemma \ref{lem:set-of-dirty-intervals} shows that set $\D$ is a superset of all intervals which became dirty after deletion of interval $I_i$.  Thus we bound the set $\D$ by showing that only those intervals which are adjacent to $I_i$ and whose level value is more than $L(I_i)$ can be dirty during the deletion of $I_i$.
%%%%%%%%%%%%%%%%%%%%%%%%%%%%%%%%%%%%%%%%%%%%%%%
\begin{lemma}
	\label{lem:set-of-dirty-intervals}
	During the execution of Algorithm $\Deletion$ on the update {\sf Delete$(I_i)$}, if an interval $I$ becomes dirty, then $I \in \D$.
%	Let $I_i$ be the deleted interval and let $(L(I_i),o(I_i))$ be the color of  $I_i$. Then the set $\D$ is a super set of all the intervals those may become dirty after the delete. 
\end{lemma} 
\begin{proof}
On deletion of $I_i$, the supporting line segments are naturally classified into two sets: those whose height  reduces and those whose height does not reduce.  Let $t$ be a point for which the height does not reduce.  First we consider an interval $I$ which contains $t$, $h_t \geq L(I)$,  and $L(I) \leq L(I_i)$, and show that $I$ does not become dirty during the execution of Algorithm $\Deletion$.  This is because, during the execution of the algorithm, only intervals with level value greater than $L(I_i)$ are considered for  a reduction in level value.  Therefore, the height of $t$ will remain at least $L(I)$ throughout the execution of the algorithm.  Therefore, $I$ does not become dirty during the execution of Algorithm $\Deletion$.  

Next, let us consider an interval $I$ which contains $t$, $h_t \geq L(I)$, and $L(I) > L(I_i)$.  
%such that $h_t$ does not reduce and $L(I) \leq L(I_i)$, it follows that $I$ does not become dirty during Algorithm $\Deletion$.  The supporting line segments of an interval whose height is at most $L(I_i)$ do not reduce by definition
%By Lemma \ref{lem:3-competitive-fully-dynamic}, we know that prior to the update  {\sf Delete$(I_i)$} invariant $\C$ is satisfied by each interval.  
%Let $I$ be an interval such that $L(I) \leq L(I_i)$.   
%As  the  level number of $I'$ can only reduce upto the maximum height of  the SLS at a point $t'$ in $I'$.  
The algorithm considers the intervals $I' \in \D$ in increasing order of level number. 
Thus, it follows that the level value of an interval $I'$ which contains $t$ and for which $L(I') \leq L(I)$ will not reduce during the execution of Algorithm $\Deletion$.  Therefore, the height of the SLS at $t$ does not change throughout the execution of Algorithm $\Deletion$, and thus $I$ does not become dirty during the execution of Algorithm $\Deletion$.  
Therefore, an interval $I$ which becomes dirty during the execution of Algorithm $\Deletion$ must contain a $t$ whose height reduces on the deletion of $I_i$.  Thus, $I$ intersects with $I_i$ at $t$, and as we have proved above, $I$ must have level value more than $L(I_i)$. In other words it must be an element of $\D$.  Hence the Lemma.
\end{proof}
%%%%%%%%%%%%%%%%%%%%%%%%%%%%%%%%%%%%%%%%%%%%%%%%%%%%%%%%%%%%%%%%%%%%%%%%%%%%%%%%%%%
\subsection{Worst-case analysis of runtime of $\Deletion$ and  $\Insertion$}
\label{subsec:analysis-fully-dynamic-algorithm}
%In this section we present a worst case analysis of $\Deletion$ (Algorithm~\ref{alg:delete}) and $\Insertion$ (Algorithm~\ref{alg:insert}).
%We analyze worst case performance of every step in  Algorithm~\ref{alg:delete} and present a worst case update time for deletion. Similarly, we analyze worst case performance of every step in $\Insertion$ and present a worst case update time for insertion. 
%%%%%%%%%%%%%%%%%%%%%%%%%%%%%%%%%%%%%%%%%%%%%%%%%%%%%%%%%%%%%%%%%%%%%%%%%%%%%%%%%%%%%
%\subsubsection{Analysis of $\Deletion$ in the fully dynamic setting}
\label{subsubsec: analysis-fully-dynamic-deletion}
%Let $I_i = [l_i,r_i]$ with color $(L(I_i), o(I_i))$  
%be the interval presented to $\Deletion$ for deletion in the current update step. Let $n$ be the total number of intervals inserted during the course of execution of the algorithm. 
%%%%%%%%%%%%%%%%%%%%%%%%%%%%%%%%%%%%%%%%%
\begin{lemma}
%	\label{lem: deletion-step4}	
\label{thm : fully-dynamic-update-time-deletion}
	Algorithm \ref{alg:delete} implements  $\Deletion$ in $O(\Delta^2 \log n)$ time.
\end{lemma}
\begin{proof}
It is clear from the description that Algorithm \ref{alg:delete} implements each of the steps of $\Deletion$. The most expensive steps in Algorihtm \ref{alg:delete} are the computation and sorting of $\D$, and updating the level values of the intervals in $\D$, if necessary, in lines 11-28.  $\D$ is computed by an intersection query to the interval tree $\mathcal{E}$ and the worst-case running time  is $O(\log n + \Delta)$, where $\Delta$ is the number of endpoints in $\mathcal{E}$ which is the  number of intervals intersecting with $I_i$. Subsequently, sorting $\D$ takes time $O(\Delta \log \Delta)$ time.  Each iteration in lines 11-28 is 
for an interval $I \in \D$, and the running time of an iteration is dominated by the iteration in lines 18-25. The number of times lines 18-25 is executed is  $O(\Delta)$ which  the number of endpoints in $S$, where $S$ is $\mathcal{E} \cap I$.  In each of these iterations, the Red-Black trees are updated to reflect the height of the corresponding supporting line segment, and this takes $O(\log n)$ time, using the fact that the number of values in each Red-Black tree is at most $\omega$.  Thus the running time of Algorithm \ref{alg:delete} is $O(\Delta^{2}  \log n)$.
Hence the Lemma.
\end{proof}
\label{subsubsec: analysis-fully-dynamic-insertion}
\noindent
%Let $I_i$ = $[l_i,r_i]$ be the interval inserted in the $i$-th update step. Running time of $\Insertion$(Algorithm~\ref{alg:insert}) in the fully dynamic setting differs from incremental setting because of the use of Red-Black trees to maintain supporting line segments instead of dynamic arrays and linked list.
We next analyze
$\Insertion$ which is implemented by Algorithm ~\ref{alg:insert}  in the fully-dynamic setting  by representing
  supporting line segments as Red-Black trees. Recall that  in the incremental case they were represented by dynamic arrays and a doubly linked list.
%%%%%%%%%%%%%%%%%%%%%%%%%%%%%%%%%%%%%%%%%%
%The update time of $\Insertion$(Algorithm~\ref{alg:insert}) is given by the following theorem.
\begin{lemma}
	\label{thm :fully-dynamic-update-time-insertion}
	Algorithm \ref{alg:insert} with supporting line segments represented as Red-Black trees implement $\Insertion$ in the fully-dynamic setting.  
	Algorithm \ref{alg:insert} inserts an interval in worst case $O(\log n + \Delta \log \omega)$ time.
%	In the fully dynamic setting, insertion of an interval takes  $O(\log n + \Delta \log \omega)$ time in the worst case.
\end{lemma}
\begin{proof}
In Lemma \ref{lem:runtimeComputeSLSfullyDynamic}, Lemma \ref{lem:runtimeComputemaxheightfullydynamic}, and Lemma \ref{lem:runtimeUpdaeSLSfullydynamic},   respectively, we prove that in the fully-dynamic setting, with SLS represented as Red-Black trees,  $\getSLS$, $\computemaxheightSLS$, and  $\updateEndPoints$ correctly implement Step 1, Step 2, and Step 3, of $\Insertion$ correctly.  Therefore, it follows that $\Insertion$ is implemented by Algorithm ~\ref{alg:insert} correctly.  Further, these Lemmas also show that the worst-case running times of these functions
is $O(\log n + \omega \log \omega)$, $O(\log n + \Delta \log \omega)$, $O(\Delta \log \omega)$, respectively.  Therefore, the worst-case running time, of
$\Insertion$ implemented by Algorithm \ref{alg:insert}, with SLS represented as Red-Black trees, is  $O(\log n + \Delta \log \omega)$.  Hence the Lemma.
\end{proof}
\noindent
%%%%%%%%%%%%%%%%%%%%%%%%%%%%%%%%%%%%%%%%%%%%%%%%%%%%%%%%%%%%%%%%%%%%%%%%%%%%%%%%%%%%%
%%%%%%%%%%%%%%%%%%%%%%%%%%%%%%%%%%%%%%%%%%
Our fully-dynamic algorithm for interval coloring follows by combining Lemma~\ref{thm : fully-dynamic-update-time-deletion} and Lemma~\ref{thm :fully-dynamic-update-time-insertion}.
\begin{theorem}
	\label{thm : update-time-fully-dynamic}
	There exists a fully dynamic algorithm which supports insertion of an interval in  $O(\log n + \Delta \log \omega)$ and deletion of an interval in  $O(\Delta^{2} \log n)$ worst case time. 
\end{theorem}
%%%%%%%%%%%%%%%%%%%%%%%%%%%%%%%%%%%
%%%%%%%%%%%%%%%%%%%%%%%%%%%%%%%
\subsection{Procedures used in $\Deletion$ and $\Insertion$}
\label{subsec:procedureHandleDeleteInsert}
\noindent
The procedures $\getSLS$, $\computemaxheightSLS$, and  $\updateEndPoints$ which are defined in Section~\ref{subssec:procedures-handle-insert} are
defined in this section with Red-Black trees used to represent supporting line segments.  The worst-case running time of these procedures differ from their running times in Section~\ref{subssec:procedures-handle-insert}.
The data structures used in designing these procedures are listed in Table~\ref{table:1}.

\noindent
%For the purpose of completeness we describe the procedures again along with running time analysis. 
%A summarized description of the procedures is given in Table~\ref{table:3}. \\
%%%%%%%%%%%%%%%%%%%%%%%%%%%%%%%%%%%%%%%%%%%%%
\begin{minipage}{0.45\textwidth}
\begin{lemma}
\label{lem:runtimeComputeSLSfullyDynamic}	
Procedure $\getSLS$ takes as input the set of intervals $\mathcal{I}$ and endpoint $t$, maintains SLS at $t$ as Red-Black trees $Z_t$ and ${NZ}_t$, and takes $O(\log n + \omega \log \omega)$ time in the worst case.
\end{lemma}
\begin{proof}
$\getSLS$ performs an intersection query on $\mathcal{I}$ with $[t,t]$ (Line 2 in Algorithm~\ref{alg:GET-SLS}). The query returns all the intervals in $\mathcal{I}$ which contain endpoint $t$. Let $\mathcal{I}_{t}$ denote the set returned by the intersection query. Set $U$ = $levels(\mathcal{I}_{t})$ and height $h_t$ = $\max(levels(\mathcal{I}_{t}))$ are computed (Line 5-8 in Algorithm~\ref{alg:GET-SLS}). For every $i$ in the range $[0,h_t]$, 
$i$ is inserted to $Z_{t}$ (Line 11-13 in Algorithm~\ref{alg:GET-SLS}). For every $i$ in the set $U$, $i$ is deleted from $Z_t$ and inserted to ${NZ}_t$ (Line 14-17 in Algorithm~\ref{alg:GET-SLS}). $\getSLS$ returns $Z_t$ and ${NZ}_t$.\par   	
Running time of $\getSLS$ is dominated by the intersection query in Line 2, and loops in Line 11-13 and Line 14-17. 	
At any level, SLS ${e_t}$ intersects with at most $2$ intervals and we have $\omega$ many levels. Hence, $|\mathcal{I}_{t}|$ = $O(\omega)$. Again, $|\mathcal{I}| \leq n$. Therefore, intersection query takes $O(\log n + \omega)$. Further, a single insertion in $Z_{t}$ and ${NZ}_{t}$ takes $O(\log \omega)$ time. Therefore, total time taken by the loops is $O(\omega \log \omega)$.
This implies that the worst case time taken by $\getSLS$ is $O(\log n + \omega \log \omega)$. Hence the Lemma.	
\end{proof}
\noindent
\begin{lemma}
\label{lem:runtimeComputemaxheightfullydynamic}	
Procedure $\computemaxheightSLS$ takes as input set of endpoints $\mathcal{E}$ and interval $I$, computes the maximum height of SLS contained in interval $I$, and takes $O(\log n + \Delta \log \omega)$ time in the worst case.
\end{lemma}
\begin{proof}
$\computemaxheightSLS$ works as follows (Algorithm~\ref{alg:MAX-HEIGHT-OF-SLS-IN-INTRVAL}): an intersection query is performed on $\mathcal{E}$ with $I$ (Line 2). Let $S$ be the set returned by the intersection query.
For every endpoint $t \in S$, to compute the height $h_t$ of SLS $e_t$, the following steps are used: If $Z_t$ is non empty then the minimum value in $Z_t$ is assigned to $h_t$ (Line 10). Otherwise, $h_t$ is assigned a value which is one more than the maximum value in ${NZ}_t$ (Line 8). The maximum value of the height of an SLS at any endpoint in $S$ is computed as $h$ = $\max \{h_t | t \in S\}$ (Line 12). The procedure returns the set $S$ and value $h$.\par  	
Running time of $\computemaxheightSLS$ is dominated by the intersection query in Line 2 and the loop in Line 4-13.  	
We know that $|S| \leq \Delta$ and $|\mathcal{E}| \leq 2n$.	
Therefore, worst case time taken by $\computemaxheightSLS$ is $O(\log n + \Delta \log \omega)$. Hence the Lemma.
\end{proof}

\end{minipage}
\hspace{0.6cm}
\begin{minipage}{0.45\textwidth}
%%%%%%%%%%%%%%%%%%%%%%%%%%%%%%%%%%%%%%%%%%%%%%%%%%%%%%
\begin{algorithm}[H]
	\caption{$\getSLS$($\mathcal{I}$,$t$) is used to compute the supporting line segment at endpoint $t$.}
	% In this procedure SLS is maintained using Red-Black Trees.}
	\begin{algorithmic}[1]
		\Procedure{$\getSLS$}{$\mathcal{I}$,$t$}
		\State $\mathcal{I}_t \leftarrow$  $\mathcal{I}$.intersection($[t,t]$)
		\State $h_t \leftarrow 0$
		\State{$U \leftarrow$ Initialize to empty set}
		\For{$I$ \textbf{in} $\mathcal{I}_t$}
		\State ${U}.$insert($L(I)$)
		\State $h_t \leftarrow max(h_t,L(I))$
		\EndFor
		\State $Z_t \leftarrow $ Empty Red-Black tree
		\State $NZ_t \leftarrow $ Empty Red-Black tree
		\For{$i$ \textbf{in} $\{0,1,2,...h_t\}$} 
		\State $Z_t.$insert($i$)
		\EndFor
		\For{$i$ \textbf{in} $U$.begin() \textbf{to} $U$.end()} 
		\State ${NZ}_t.$insert($i$)
		\State $Z_t.$delete($i$)
		\EndFor
		\State \Return $Z_t, {NZ}_t$
		\EndProcedure
	\end{algorithmic}
	\label{alg:GET-SLS}
\end{algorithm}

\begin{algorithm}[H]
	\caption{$\computemaxheightSLS$($\mathcal{E}$,$I$) is used to compute the set $S$ of endpoints contained in interval $I$ and compute the maximum height $h$ of all the supporting line segments at these endpoints.}
	% In this procedure SLS is maintained using Red-Black Trees.}
	\begin{algorithmic}[1]
		\Procedure{$\computemaxheightSLS$}{$\mathcal{E}$,$I$} 
		\State $S \leftarrow$  $\mathcal{E}$.intersection($I$)
		\State $h \leftarrow 0$
		\For{$t$ \textbf{in} $S$}
		\If{$Z_t$.empty() $\And$ ${NZ}_t$.empty()}
		\State $h_t \leftarrow $ 0
		\ElsIf{$Z_t$.empty()}
		\State $h_t \leftarrow {{NZ}_t}.\max()+1$
		\Else 
		\State $h_t \leftarrow {Z_t}.\min()$
		\EndIf
		\State $h \leftarrow \max(h,h_t)$
		\EndFor
		\State return $S, h$
		\EndProcedure
	\end{algorithmic}
	\label{alg:MAX-HEIGHT-OF-SLS-IN-INTRVAL}
\end{algorithm}	
\end{minipage}

%\begin{minipage}{0.5\textwidth}
%\end{minipage}
%\hspace{0.6cm}
%\begin{minipage}{0.5\textwidth}
%%%%%%%%%%%%%%%%%%%%%%%%%%%%%%%%%%%%%%%%%%%
%\end{minipage}

\noindent
%%%%%%%%%%%%%%%%%%%%%%%%%%%%%%%%%%%%%%%%%%%%%%%%%
%%%%%%%%%%%%%%%%%%%%%%%%%%%%%%%%%%%%%%%%%%%%%%%%% 
%%%%%%%%%%%%%%%%%%%%%%%%%%%%%%%%%%%%%%%%%%%%%%%%%
\begin{minipage}{.45\textwidth}
\begin{lemma}
\label{lem:runtimeUpdaeSLSfullydynamic}	
Procedure $\updateEndPoints$ takes set of endpoints $S$ and $L(I)$ as input, update the SLS at the endpoints contained in set $S$ and takes $O(\Delta \log \omega)$ time in the worst case.
\end{lemma}
\begin{proof}
The procedure works as follows (Algorithm~\ref{algo:UPDATE-EDGE-POINTS}): 
for every $ t \in S$, $L(I)$ is deleted from $Z_t$ and $L(I)$ is inserted to ${NZ}_t$. For one SLS $e_t$ it takes $O(\log \omega)$ time and $|S| \leq \Delta$. Therefore, worst case time taken by $\updateEndPoints$ is $O(\Delta \log \omega)$. Hence the Lemma.
\end{proof}
\textbf{Procedure $\OFFSET$$(I)$:} This procedure is same as the one described in Section~\ref{subssec:procedures-handle-insert}.
\end{minipage}
\hspace{0.6cm}
%%%%%%%%%%%%%%%%%%%%%%%%%%%%%%%%%%%%%%%%%%%%%%%%%%%
\begin{minipage}{0.45\textwidth}
\begin{algorithm}[H]
	\caption{$\updateEndPoints$($S,L(I)$) is used to update the supporting line segments at the endpoints contained in set $S$ for level value $L(I)$. }
	%In this procedure SLS is maintained using Red-Black Trees.}
	\begin{algorithmic}[1]
		\Procedure{$\updateEndPoints$}{$S$,$L(I)$} 
		\For{$t$ \textbf{in} $S$}
		\If{$ \neg Z_t$.empty()}
		\State{$h_t \leftarrow {Z_t}.min()$}
		\EndIf
		\If{$Z_t$.empty()}
		\State{$h_t \leftarrow {{NZ}_t}.max() + 1$}
		\EndIf
		\State{$q$ = $h_t$}    
		\While{$q < L(I)$} 
		\State ${Z_t}.$insert$(q)$
		\State{ $q \leftarrow q + 1$} 
		\EndWhile
		\If{$L(I)$ \textbf{in} $Z_t$} 
		\State ${Z_t}.$delete$(L(I))$ 
		\EndIf
		\State ${{NZ}_t}.$insert$(L(I))$
		\EndFor
		\EndProcedure
	\end{algorithmic}
	\label{algo:UPDATE-EDGE-POINTS}
\end{algorithm}
\end{minipage}
%%%%%%%%%%%%%%%%%%%%%%%%%%%%%%%%%%%%%%%%%%%%%%%%%%%
%%%%%%%%%%%%%%%%%%%%%%%%%%%%%%%%%%%%%%%%%%%%%%%%%%%%%%%
%%%%%%%%%%%%%%%%%%%%%%%%%%%%%%%%%%%%%%%%%%%%%%%%%%%%%%%%%%%%%
%%%%%%%%%%%%%%%%%%%%%%%%%%%%%%%%%%%%%%%%%%%%%%%%%%%% 
%%%%%%%%%%%%%%%%%%%%%%%%%%%%%%%%%%%%%%%%%%%%%%%% 
%%%%%%%%%%%%%%%%%%%%%%%%%%%%%%%%%%%%%

%% file: Journal/OMV.tex
% !TEX root = ../IntervalColoringFullVersion.tex
\section{Quadratic lower bound for induced neighborhood subgraph computation}
\label{sec:OMv}
\noindent
%As discussed in Section~\ref{KTO}, in the online setting for the interval coloring problem, when the $i$-th interval in the online sequence, $I_i$, is presented to the KT-algorithm for coloring, KT-algorithm does the following: it considers the interval graph $G$ formed by the intervals $\{I_1, I_2, \dots, I_{i-1}\}$ and computes the induced subgraph among the neighbors of $v_i$, where $v_i$ is the vertex corresponding to interval $I_i$ in $G$. This is a very crucial step in the KT-algorithm to compute the level value for an interval $I_i$. Our result in Section~\ref{sec:KTalgorithm-hardness} shows that this step stands as a barrier in obtaining a sub-quadratic time implementation of the KT-algorithm. This motivates us to consider the problem of computing induced subgraph among the neighbors of a given vertex in a general graph and investigate its complexity.\par 
%%%%%%%%%%%%%%%%%%%%%%%%%%%%%%%%%%%%%%%%%%%%%%%%%
In Section~\ref{sec:KTalgorithm-hardness} we showed that a direct implementation of the KT-algorithm will not run in  sub-quadratic time.  
From Section \ref{KTO} the crucial step is to compute maximum clique in an induced subgraph of the neighborhood of the interval inserted during an update.   In this section we explore an interesting connection between computing the induced subgraph of the neighborhood of a vertex in a graph and the well-known OMv conjecture due to Henzinger et al., \cite{DBLP:conf/stoc/HenzingerKNS15}.
%The subsequent results in Section \ref{} avoided this step by maintaining candidate cliques and querying them efficiently.
%In this section we show that in a general the problem of computing the induced subgraph of the closed neighborhood of a set of vertices is unlikely to have a sub-quadratic time algorithm. 
Formally, we define the following problem:\\
{\bf Induced Neighborhood Subgraph Computation: } The input to the {\em Induced Neighborhood Subgraph Computation} problem consists of the adjacency matrix $M$ of a directed graph and a set $S$ of vertices.  The goal is to compute the graph induced by $N_{out}(S) \cup S$ and output the subgraph as adjacency lists. Here $N_{out}(S)$ is the set of those vertices which have  a directed edge from some vertex in $S$.  In other words, there is a directed edge from $v_j$ to $v_k$ iff the entry $M[k][j]$ is $1$.\par 
\noindent
We show that Induced Neighborhood Subgraph Computation problem is at least as hard as the following problem.\\
\textbf{Online Boolean Matrix-Vector Multiplication (OMv)}\cite{DBLP:conf/stoc/HenzingerKNS15}: The input for this online problem consists of an $n \times n$ matrix $M$, and a sequence of $n$ boolean column vectors $v_1, \ldots, v_n$, presented one after another to the algorithm.  For each $1 \leq i \leq n-1$, the online algorithm should output $M \cdot v_i$ before $v_{i+1}$ is presented to the algorithm.  Note that in this product, a multiplication is an AND operation and the addition is an OR operation.\par
\noindent
%%%%%%%%%%%%%%%%%%%%%%%%%%%%%%%%%%%%%%%%%%%%%%%%
The current best algorithm for the OMv problem has an expected running time of $O(\frac{n^3}{2^{\sqrt{\log n}}})$ \cite{DBLP:conf/soda/LarsenW17}. The following conjecture, due to Henzinger et al., \cite{DBLP:conf/stoc/HenzingerKNS15}, is well known about the OMv problem.\\
{\bf OMv conjecture: } The Online Boolean Matrix-Vector Multiplication (OMv) problem does not have a $O(n^{3 - \epsilon})$ algorithm  for any $\epsilon > 0$.\par  
%%%%%%%%%%%%%%%%%%%%%%%%%%%%%%%%%%%%%%%%%%%%%%%%%%%%%%%
\noindent
In Theorem \ref{thm:InducedSubGraphviaOMv} we reduce OMv problem to Induced Neighborhood Subgraph Computation problem. As a conseuqence of our reduction an efficient algorithm for Induced Neighborhood Subgraph Computation problem implies an efficient algorithm for the OMv problem. 
%%%%%%%%%%%%%%%%%%%%%%%%%%%%%%%%%%%%%%%%%%%%%%%%%
\begin{theorem}
\label{thm:InducedSubGraphviaOMv}	
Any algorithm for Induced Neighborhood Subgraph Computation problem needs at least quadratic time unless OMv conjecture is false.
\end{theorem}
\begin{proof}
We show that an algorithm to solve the Induced Neighborhood Subgraph Computation problem can be used to solve the Online Boolean Matrix-Vector Multiplication problem.  
Let $\mathcal{A}$ be an algorithm for the Induced Neighborhood Subgraph Computation problem  with the running time of $\mathcal{A}$ being $O(n^{2-\epsilon})$, for some $\epsilon > 0$. We use algorithm $\mathcal{A}$ to solve the Online Boolean Matrix-Vector Multiplication problem in $O(n^{3-\epsilon})$ time as follows : Let $M$ be the input matrix for the Online Boolean Matrix-Vector Multiplication problem and let $V_1, \ldots V_n$ be the column vectors presented to the algorithm one after the other.  For the column vector $V_i$, let set $S_i = \{v_j | V_i[j] = 1, 0\leq j\leq n-1\}$.  To compute $M \cdot V_i$, we invoke $\mathcal{A}$ on input $\{M, S_i\}$. Let $G_{S_i}$ denote the induced subgraph  on $N_{out}(S_i) \cup S_i \subseteq V$ computed by the algorithm $\mathcal{A}$.  
Note that $G_{S_i}$ is an induced subgraph of the directed graph whose adjacency matrix is $M$. To output the column vector $M \cdot V_i$, we observe that the $j$-th row in the output column vector is 1 if and only if $v_j \in G_{S_i}$ and there is an edge $(u, v_j)$ in $G_{S_i}$ such that $u \in S_i$. %Therefore, 
Given that $G_{S_i}$ has been computed in $O(n^{2-\epsilon})$ time, it follows that the number of edges in $G_{S_i}$ is $O(n^{2-\epsilon})$ and consequently the column vector $M \cdot V_i$ can be computed in $O(n^{2 - \epsilon})$ time.   
Therefore, using the $O(n^{2-\epsilon})$ algorithm $\mathcal{A}$ we can solve Boolean Matrix-Vector Multiplication problem in $O(n^{3 - \epsilon})$ time.  %Therefore, 
If we believe that the OMv conjecture is indeed true, then it follows that the Induced Neighborhood Subgraph Computation problem cannot have an $O(n^{2-\epsilon})$ algorithm for any $\epsilon > 0$. Hence the Theorem.
\end{proof}  
%%%%%%%%%%%%%%%%%%%%%%%%%%%%%%%%%%%%%%%%%%%%%%%%
%Given an interval $I = [l,r]$, $I$ can be interpreted as a vector $v$ of length $n$, where $n \geq r$, and the entries of vector $v$ as $v[i][1] = 1$ 
%for $l \leq i \leq r$ and $0$ in all other positions. In other words, any interval $I = [l,r]$ can be visualized as a vector of length $n \geq r$ with consecutive ones from position $l$ to $r$. Using this interpretation we present our final result in Section~\ref{subsec:OMvWithConsecutiveOnes} for a special case of the OMv problem where the input matrix and the vectors in the online sequence have consecutive ones property.
%%%%%%%%%%%%%%%%%%%%%%%%%%%%%%%%%%%%%%%%%%%%%%%%
\subsection{OMv conjecture is false for instances with the consecutive ones property}
\label{subsec:OMvWithConsecutiveOnes}
\noindent
A 0-1 matrix is said to have the consecutive ones property if in each row, the column indices which have a 1 form an interval.  
A  0-1 column vector satisfies the consecutive ones property if the row indices which have a 1 in the column form an interval.
We consider a special case of the OMv problem where the input matrix $M$ and the sequence of online vectors $\{v_1, v_2, \dots, v_n\}$ satisfy consecutive ones property.   Each row in the matrix $M$ corresponds to an  interval and every column index is a point on the number line.
%For simplicity of presentation, we assume that $M$ is a square matrix of size $n \times n$ and hence every vector in the sequence is of size $n \times 1$. We use the following interpretation.\par
%%%%%%%%%%%%%%%%%%%%%%%%%%%%%%%%%%%%%%%%%%%% 
In particular, in the $i$-th row if  $l$ and $r$ are the least and largest column index, respectively, such that $M[i][l] = M[i][r] = 1$,  then the $i$-th row corresponds to the  interval $I_i = [l,r]$.
 %Further, we interpret every vector $v_j$, where $j$ is in the range $1$ to $n$, in the online vector sequence as an interval.
  For each $1 \leq j \leq n$, if $p$ and $q$ are the least and largest indices in $v_j$ such that $v_j[p] = v_j[q] = 1$ then the vector $v_j$ is interpreted as the interval $I_{v_j} = [p, q]$.\par
  \noindent
  Now, using the data structures described in Table \ref{table:1} in Section~\ref{subsec:DSpreliminaries} we design an algorithm to solve OMv problem in quadratic time for this special case. 
%%%%%%%%%%%%%%%%%%%%%%%%%%%%%%%%%%%%%%%%%%%%%%%%
\begin{theorem}
\label{thm:OMvConjConsOnes}
OMv conjecture is false if the input matrix and the vectors in the online vector sequence have the consecutive ones property.
\end{theorem}
%%%%%%%%%%%%%%%%%%%%%%%%%%%%%%%%%%%%%%%%%%%%%%%%
\begin{proof}
%
%We design a quadratic time algorithm for the special case of the OMv problem where the input matrix $M$  and the vectors in the online vector sequence $\{v_1, v_2, \dots, v_n\}$ satisfy consecutive ones property. Thus we prove that OMv conjecture is false for this special case of the OMv problem.\par
The proof is by presenting an algorithm to solve the OMv problem.  The algorithm has a preprocessing step in which the intervals corresponding to the rows of the matrix are maintained in an interval tree.   Subsequently, the matrix vector product is computed using queries to the interval tree. \\
{\bf Preprocessing step: } The interval corresponding to the rows in $M$ are computed.
 Let $\R$ = $\{I_1, I_2, \dots, I_n\}$ denote the set of intervals corresponding to the rows in $M$. Computing the  set $\R$ takes $O(n^2)$ time. An interval tree $T_{\R}$ is constructed using the set $\R$. The construction of $T_{\R}$ takes $O(n \log n)$ time. Therefore, total time required in the preprocessing step is  $O(n^2)$. \\
%%%%%%%%%%%%%%%%%%%%%%%%%%%%%%%%%%%%%
%After the preprocessing step, vectors in the sequence are presented one by one to the algorithm.\\
{\bf Computing $M \cdot v_j$: } For $1 \leq j \leq n$, when vector $v_j$ is presented, interval $I^j$ corresponding to $v_j$ is computed in $O(n)$ time. 
Let $H$ be the set returned by the intersection query $T_{\R}.intersection(I^j)$. Since, $|H| \leq n$, from Table \ref{table:1} the time required by the query is $O(\log n + n)$.
% Let $v_{\out}$ denote the output vector of the multiplication $M \cdot v_j$. 
$v_{\out}
=  M \cdot v_j$ is now computed as follows: for each $1 \leq i \leq n$, if interval $I_i$ is present in $H$ then the $i$-th position in the vector $v_{\out}$ is set to $1$, otherwise $0$. Thus $v_{\out}$ can be computed in $O(n)$ time. Therefore, total time required for computing $M \cdot v_j$ is $O(n)$. Thus the OMv problem on such instances can be solved in time $O(n^2)$. Hence the Theorem. 
%
%%%%%%%%%%%%%%%%%%%%%%%%%%%%%%%%%%%%%%%
%Therefore, total time required by the algorithm to solve the special case of the OMv problem is the sum total of the time required for the preprocessing step and the time required for computing $M \cdot v_j$ for every $j$ in the range $1$ to $n$. That is, $O(n^2)$ + $n \times O(n)$ = $O(n^2)$.\par
%%%%%%%%%%%%%%%%%%%%%%%%%%%%%%%%%%%%%%%%%%%
%Therefore, if the input matrix and the vectors in the online sequence have consecutive  ones property then OMv problem can be solved in quadratic time falsifying the OMv conjecture. 
\end{proof}
\vspace{1cm}